\documentclass[10pt,fullpage]{article}

\usepackage{fullpage}
\usepackage{times}
\usepackage{amsmath}
\usepackage{amssymb}
\usepackage{amstext}
\usepackage{amsthm}
\usepackage{epsfig}
\usepackage{graphicx}
\usepackage{placeins}
\usepackage{url}
\usepackage{subfig}
\usepackage{paralist}
\usepackage{color}
\usepackage{caption}
\usepackage{multirow}
\usepackage{enumitem}
\usepackage{algorithm}
\usepackage{algorithmic}
\usepackage{comment}

\newcommand{\eat}[1]{}

\makeatletter
\newenvironment{sql}%
 {\vskip 5pt\begin{list}{}{%
  \setlength{\topsep}{0pt}\setlength{\partopsep}{0pt}\setlength{\parskip}{0pt}%
  \setlength{\parsep}{0pt}\setlength{\labelwidth}{0pt}%
  \setlength{\rightmargin}{0pt}\setlength{\leftmargin}{0pt}%
  \setlength{\labelsep}{0pt}%
  \obeylines\@vobeyspaces\normalfont\ttfamily%
  \item[]}}
 {\end{list}\vskip5pt\noindent}
\makeatother

\newtheorem{ex}{Example}
\newtheorem{thm}{Theorem}

\newcommand{\RED}[1]{ {\color{red}#1} }

\begin{document}

\date{January 2017}

\title{Dot-Product Join: An Array-Relation Join Operator for\\
Big Model Analytics}

\author{
Chengjie Qin \hspace*{2cm} Florin Rusu\\
       \small{University of California Merced}\\
       \small{5200 N Lake Road}\\
       \small{Merced, CA 95343}\\
       \small\texttt{\{cqin3, frusu\}@ucmerced.edu}
}

\maketitle

\begin{abstract}

Big Model analytics tackles the training of massive models that go beyond the available memory of a single computing device, e.g., CPU or GPU. It generalizes Big Data analytics which is targeted at how to train memory-resident models over out-of-memory training data.
In this paper, we propose an in-database solution for Big Model analytics. We identify dot-product as the primary operation for training generalized linear models and introduce the first array-relation dot-product join database operator between a set of sparse arrays and a dense relation. This is a constrained formulation of the extensively studied sparse matrix vector multiplication (SpMV) kernel. The paramount challenge in designing the dot-product join operator is how to optimally schedule access to the dense relation based on the non-contiguous entries in the sparse arrays. We prove that this problem is NP-hard and propose a practical solution characterized by two technical contributions---dynamic batch processing and array reordering. We devise three heuristics -- LSH, Radix, and K-center -- for array reordering and analyze them thoroughly. We execute extensive experiments over synthetic and real data that confirm the minimal overhead the operator incurs when sufficient memory is available and the graceful degradation it suffers as memory becomes scarce. Moreover, dot-product join achieves an order of magnitude reduction in execution time over alternative in-database solutions.

\end{abstract}

\section{Introduction}\label{sec:intro}

Big Data analytics is a major topic in contemporary data management and machine learning research and practice. Many platforms, e.g., OptiML~\cite{optiml}, GraphLab~\cite{graphlab,graphlab-disk,graphlab-distributed}, SystemML~\cite{systemml}, Vowpal Wabbit~\cite{vowpal-wabbit}, SimSQL~\cite{simsql}, GLADE~\cite{glade:sigmod} and libraries, e.g., MADlib \cite{mad-skills,madlib}, Bismarck~\cite{bismarck}, MLlib~\cite{mllib}, Mahout\footnote{\url{https://mahout.apache.org}}, have been proposed to provide support for distributed/parallel statistical analytics. We can categorize these solutions into general frameworks with machine learning support -- Mahout, Spark's MLLib, GraphLab -- dedicated machine learning systems -- SystemML, SimSQL, OptiML, Vowpal Wabbit -- and frameworks within databases---MADlib, Bismarck, GLADE. In this paper, we focus on the last category---frameworks for in-databse analytics. A common assumption across all these systems is that the number of model parameters or features is small enough to fit in memory. This is made explicit by the representation of the model as an in-memory array data structure. However, due to the explosive growth in data acquisition and the wide adoption of analytics methods, the current trend is to devise models with an ever-increasing number of features---\textit{Big Models}. A report on industrial machine learning\footnote{\url{http://www.kdnuggets.com/2014/08/sibyl-google-system-large-scale-machine-learning.html}} cites models with 100 billion features (800 GB in size) as early as 2012. Scientific applications such as ab initio nuclear structure calculations also generate extremely large models with billions of features~\cite{ab-initio}. While these are extreme cases, there are many realistic applications that require Big Models and are forced to limit the number of features they consider because of insufficient memory resources. We provide two such examples in the following.

\textbf{Example 1: Recommender systems.}
A class of analytics models with highly-dimensional feature space are the ones in which the dimensionality grows with the number of observations. Low-rank matrix factorization (LMF) is a typical example. In LMF, the observations are a set of cells in a sparse $m \times n$ matrix $M$. The non-empty cells represent the users' ratings of items in a certain category -- such as songs or movies -- with each row corresponding to a user and each column to an item. In general, every row is sparse since a typical user rates only a very limited set of items. Each column is also sparse since only a restricted number of users rate an item. LMF seeks to decompose matrix $M$ into the product of two dense low-rank matrices $L$ and $R$ with dimensions $m \times r$ and $r \times n$, respectively, where $r$ is the rank. The prediction accuracy increases with $r$. The LMF features are matrices $L$ and $R$ which grow with the number of users $m$ and the number of items $n$, respectively, and the rank $r$. LMF is heavily used in recommender systems, e.g., Netflix, Pandora, Spotify. For example, Spotify applies LMF for 24 million users and 20 million songs\footnote{\url{http://www.slideshare.net/MrChrisJohnson/algorithmic-music-recommendations-at-spotify}}, which leads to 4.4 billion features at a relatively small rank of 100.

\textbf{Example 2: Topic modeling.}
n-grams are a common feature vector in topic modeling. They are extracted by considering sequences of 1-word tokens (unigrams), 2-word tokens (bigrams), up to n-word tokens (n-grams) from a fixed dictionary. The feature vector consists of the union of unigrams, bigrams, $\dots$, n-grams. Several analytics models can be applied over this feature space, including latent Dirichlet allocation, logistic regression, and support vector machines. For the English Wikipedia corpus, a feature vector with 25 billion unigrams and 218 billion bigrams can be constructed~\cite{Lee:STRADS}. A similar scale can be observed in genomics where topic modeling is applied to genome sequence analysis.

\textbf{Existing solutions.}
The standard model representation across all the Big Data analytics systems we are aware of -- in-database or not -- is a memory-resident container data structure, e.g., \texttt{vector} or \texttt{map}. Depending on the parallelization strategy, there can be one or more model instances in the system at the same time. Hogwild!~\cite{nolock-igd} uses a single non-synchronized instance, while averaging techniques~\cite{igd-average-gradient} replicate the model for each execution thread. At the scale of models introduced above, a memory-resident solution incurs prohibitive cost---if it is feasible at all. In reality, in-database analytics frameworks cannot handle much smaller models. For example, MADlib and Bismarck are built using the UDF-UDA\footnote{\url{http://www.postgresql.org/docs/current/static/xaggr.html}} functionality available in PostgreSQL. The model is stored as an array attribute in a single-column table. PostgreSQL imposes a hard constraint of 1 GB for the size of an attribute, effectively limiting the model size. High performance computing (HPC) libraries for efficient sparse linear algebra such as Intel MKL\footnote{\url{https://software.intel.com/en-us/intel-mkl}}, Trilinos\footnote{\url{https://trilinos.org/}}, CUSPARSE\footnote{\url{https://developer.nvidia.com/cusparse}}, and CUSP\footnote{\url{https://github.com/cusplibrary/cusplibrary}} are optimized exclusively for in-memory processing, effectively requiring that both the training dataset and the model fit in memory simultaneously.

Two approaches are possible to cope with insufficient memory---secondary storage processing and distributed memory processing. In secondary storage processing, the model is split into partitions large enough to fit in memory and the goal is to optimize the access pattern in order to minimize the number of references to secondary storage. This principle applies between any two layers of the storage hierarchy---cache and memory, memory and disk (or SSD), and texture memory and global memory of a GPU. While we are not aware of any secondary storage solution for data analytics, there have been several attempts to optimize the memory access pattern of the SpMV kernel. However, they are targeted at minimizing the number of cache misses~\cite{phi-spmv,analytics-spmv,ssd-spmv} or the number of accesses to the GPU global memory~\cite{spmv-vldb}---not the number of disk accesses.

In distributed memory processing, the Big Model is partitioned across several machines, with each machine storing a sufficiently small model partition that fits in its local memory. Since remote model access requires data transfer, the objective in distributed processing is to minimize the communication between machines. This cannot be easily achieved for the SpMV kernel due to the non-clustered access pattern. Distributed Big Data analytics systems built around the Map-Reduce computing paradigm and its generalizations, e.g., Hadoop and Spark, require several rounds of repartitioning and aggregation due to their restrictive communication pattern. To the best of our knowledge, Parameter Server~\cite{parameter-server} is the only distributed memory analytics system capable of handling Big Models directly. In Parameter Server, the Big Model can be transfered and replicated across servers. Whenever a model entry is accessed, a copy is transferred over the network and replicated locally. Modifications to the model are pushed back to the servers asynchronously. The communication has to be implemented explicitly and optimized accordingly. While Parameter Server supports Big Models, it does so at the cost of a significant investment in hardware and with considerable network traffic. Our focus is on cost-effective single node solutions.

\textbf{Approach \& contributions.}
In this paper, we propose an in-database solution for Big Model analytics. The main idea is to offload the model to secondary storage and leverage database techniques for efficient training. The model is represented as a table rather than as an array attribute. This distinction in model representation changes fundamentally how in-database analytics tasks are carried out. We identify \textit{dot-product} as the most critical operation affected by the change in model representation. Our central contribution is the first \textit{dot-product join physical database operator} optimized to execute secondary storage array-relation dot-products effectively. Dot-product join is a constrained instance of the SpMV kernel~\cite{ssd-spmv} which is widely-studied across many computing areas, including HPC, architecture, and compilers. The paramount challenge we have to address is how to optimally schedule access to the dense relation -- buffer management -- based on the non-contiguous feature entries in the sparse arrays. The goal is to minimize the overall number of secondary storage accesses across all the sparse arrays. We prove that this problem is NP-hard and propose a practical solution characterized by two technical contributions. The first contribution is to handle the sparse arrays in \textit{batches with variable size}---determined dynamically at runtime. The second contribution is a \textit{reordering strategy} for the arrays such that accesses to co-located entries in the dense relation can be shared.

Our specific contributions can be summarized as follows:
\begin{compactitem}
\item We investigate the Big Model analytics problem and identify dot-product as the critical operation in training generalized linear models (Section~\ref{sec:problem}). We also establish a direct correspondence with the well-studied sparse matrix vector (SpMV) multiplication problem.
\item We present several alternatives for implementing dot-product in a relational database and discuss their relative advantages and drawbacks (Section~\ref{sec:baseline}).
\item We design the first array-relation dot-product join database operator targeted at secondary storage (Section~\ref{sec:dot-product}).
\item We prove that identifying the optimal access schedule for the dot-product join operator is NP-hard (Section~\ref{ssec:reordering:np-hard}) and introduce two optimizations -- dynamic batch processing and reordering -- to make the operator practical.
\item We devise three batch reordering heuristics -- LSH, Radix, and K-center (Section~\ref{sec:reordering}) -- inspired from optimizations to the SpMV kernel and evaluate them thoroughly.
\item We show how the dot-product join operator is integrated in the gradient descent optimization pipeline for training generalized linear models (Section~\ref{sec:dp-gd}).
\item We execute an extensive set of experiments that evaluate each sub-component of the operator and compare our overall solution with alternative dot-product implementations over synthetic and real data (Section~\ref{sec:experiments}). The results show that dot-product join achieves an order of magnitude reduction in execution time over alternative in-database solutions.
\end{compactitem}

\section{Preliminaries}\label{sec:problem}

In this section, we provide a brief introduction to gradient descent as a general method to train a wide class of analytics models. Then, we give two concrete examples -- logistic regression and low-rank matrix factorization -- that illustrate how gradient descent optimization works. From these examples, we identify \textit{vector dot-product} as the primary operation in gradient descent optimization. We argue that existing in-database solutions cannot handle Big Model analytics because they do not support secondary storage dot-product. We provide a formal statement of the research problem studied in this paper and conclude with a discussion on the relationship with SpMV.

\subsection{Gradient Descent Optimization}\label{ssec:problem:grad-descent}

Consider the following optimization problem with a linearly separable objective function:
\begin{equation}\label{eq:optim-form}
\Lambda(\vec{w}) = \textit{min}_{w \in \mathbb{R}^{d}} \sum_{i=1}^{N} f\left(\vec{w}, \vec{x_{i}}; y_{i}\right)
\end{equation}
in which a $d$-dimensional vector $\vec{w}$, $d \geq 1$, known as the model, has to be found such that the objective function is minimized. The constants $\vec{x_{i}}$ and $y_{i}$, $1 \leq i \leq N$, correspond to the feature vector of the $\text{i}^{\text{th}}$ data example and its scalar label, respectively.

Gradient descent represents -- by far -- the most popular method to solve the class of optimization problems given in Eq.(\ref{eq:optim-form}). Gradient descent is an iterative optimization algorithm that starts from an arbitrary model $\vec{w}^{(0)}$ and computes new models $\vec{w}^{(k+1)}$, such that the objective function, a.k.a., the loss, decreases at every step, i.e., $f(w^{(k+1)}) < f(w^{(k)})$. The new models $\vec{w}^{(k+1)}$ are determined by moving along the opposite $\Lambda$ gradient direction. Formally, the $\Lambda$ gradient is a vector consisting of entries given by the partial derivative with respect to each dimension, i.e., $\nabla \Lambda(\vec{w}) = \left[\frac{\partial\Lambda(\vec{w})}{\partial{w_{1}}}, \dots, \frac{\partial\Lambda(\vec{w})}{\partial{w_{d}}}\right]$. The length of the move at a given iteration is known as the step size---denoted by $\alpha^{(k)}$. With these, we can write the recursive equation characterizing the gradient descent method:
\begin{equation}\label{eq:grad-step}
\vec{w}^{(k+1)} = \vec{w}^{(k)} - \alpha^{(k)} \nabla \Lambda\left(\vec{w}^{(k)}\right)
\end{equation}
In batch gradient descent (BGD), the update equation is applied as is. This requires the exact computation of the gradient---over the entire training dataset. To increase the number of steps taken in one iteration, stochastic gradient descent (SGD) estimates the $\Lambda$ gradient from a subset of the training dataset. This allows for multiple steps to be taken in one sequential scan over the training dataset---as opposed to a single step in BGD. \eat{Notice that the model update is only applied to indexes with non-zero  gradient which correspond to the non-zero indexes in the training example. In order to achieve convergence, each of the subsets has to be a random sample.}

In the following, we provide two illustrative examples that show how gradient descent works for two popular analytics tasks---logistic regression (LR) and low-rank matrix factorization (LMF).

\textbf{Logistic regression.}
The LR objective function is:
\begin{equation}\label{eq:lr}
\Lambda_{\textit{LR}}(\vec{w}) = \sum_{i=1}^{N}{\log\left(1+e^{-y_{i} \vec{w} \cdot \vec{x_{i}}}\right)} 
\end{equation}
The model $\vec{w}$ that minimizes the cumulative log-likelihood across all the training examples is the solution. The key operation in this formula is the dot-product between vectors $\vec{w}$ and $\vec{x_{i}}$, i.e., $\vec{w} \cdot \vec{x_{i}}$. This dot-product also appears in each component of the gradient:
\begin{equation}\label{eq:lr-gradient}
\frac{\partial\Lambda_{\textit{LR}}(\vec{w})} {\partial{w_{i}}} = \sum_{i=1}^{N}{\left(-y_{i} \frac{e^{-y_{i} \vec{w} \cdot \vec{x}_{i}}}{1+e^{-y_{i} \vec{w} \cdot \vec{x}_{i}}} \right) \vec{x_{i}}}
\end{equation}

\textbf{Low-rank matrix factorization.}
The general form of the LMF objective function is:
\begin{equation}\label{eq:lmf}
\Lambda_{\textit{LMF}}(L,R) = \sum_{(i,j)\in M}{ \frac{1}{2} \left( {\vec{L}_{i}^{T}} \cdot \vec{R}_{j} - M_{ij} \right) ^{2}}
\end{equation}
Two low-rank matrices $L_{n \times k}$ and $R_{k \times m}$, i.e., the model, have to be found such that the sum of the cell-wise least-squares difference between the data matrix $M_{n \times m}$ and the product $L \cdot R$ is minimized. $k$ is the rank of matrices $L$ and $R$. The LMF gradient as a function of row $i'$ in matrix $\vec{L}$ is shown in the following formula:
\begin{equation}\label{eq:lmf-gradient}
\frac{\partial\Lambda_{\textit{LMF}}(L,R)} {\partial \vec{L}_{i'}} = \sum_{(i',j)\in M}{ \left( {\vec{L}_{i'}^{T}} \cdot \vec{R}_{j} - M_{i'j} \right) \vec{R}_{j}^{T}}
\end{equation}
There is such a dot-product between row $\vec{L}_{i'}$ and each column $\vec{R}_{j}$, i.e., ${\vec{L}_{i'}^{T}} \cdot \vec{R}_{j}$. However, only those dot-products for which there are non-zero cells $(i',j)$ in matrix $M$ have to be computed. A similar gradient formula can be written for column $j'$ in $\vec{R}$. 

From these examples, we identify dot-product as the essential operation in gradient descent optimization. With almost no exception, dot-product appears across all analytics tasks. In LR and LMF, a dot-product has to be computed between the model $\vec{w}$ and each example $\vec{x}_{i}$ in the training dataset---at each iteration.

\subsection{Vector Dot-Product}\label{ssec:problem:dot-product}

\begin{minipage}{.6\textwidth}
Formally, the dot-product of two $d$-dimensional vectors $\vec{u}$ and $\vec{v}$ is a scalar value computed as:
\begin{equation}\label{eq:dot-product-def}
\vec{u} \cdot \vec{v} = \sum_{i=1}^{d} {u_{i} v_{i}}
\end{equation}
This translates into a simple algorithm that iterates over the two vectors synchronously, computes the component-wise product, and adds it to a running sum (Algorithm~\ref{alg:dp-form}). In most cases, the $\textit{multiply}$ function is a simple multiplication. However, in the case of LMF, where $u_{i}$ and $v_{j}$ are vectors themselves, $\textit{multiply}$ is itself a dot-product.
\end{minipage}\hfill
\begin{minipage}{.35\textwidth}
	\centering
\begin{algorithm}[H]
\caption{Dot-Product ($\vec{u}_{1:d}$, $\vec{v}_{1:d}$)}
\label{alg:dp-form}

\algsetup{linenodelimiter=.}

\begin{algorithmic}[1]

\ENSURE $\vec{u} \cdot \vec{v}$

\STATE $\textit{dpSum} \leftarrow 0$
\FOR{$i=1$ \textbf{to} $d$}
	\STATE $\textit{dp} \leftarrow \textit{multiply} (u_{i}, v_{i})$
	\STATE $\textit{dpSum} \leftarrow \textit{dpSum} + \textit{dp}$
\ENDFOR
\RETURN {$\textit{dpSum}$}

\end{algorithmic}
\end{algorithm}
\end{minipage}\hfill

\subsection{Problem Statement \& Challenges}\label{ssec:problem:statement}

In this paper, we focus on \textit{Big Model dot-product}. Several aspects make this particular form of dot-product an interesting and challenging problem at the same time. First, the vector corresponding to the model -- say $\vec{v}$ -- cannot fit in the available memory. Second, the vector corresponding to a training example -- say $\vec{u}$ -- is sparse and the number of non-zero entries is a very small fraction from the dimensionality of the model. In these conditions, Algorithm \textit{Dot-Product} becomes highly inefficient because it has to iterate over the complete model even though only a small number of entries contribute to the result. The third challenge is imposed by the requirement to support Big Model dot-product inside a relational database where there is no notion of order---the relational model is unordered. The only solution to implement Algorithm \textit{Dot-Product} inside a database is through the UDF-UDA extensions which process a single tuple at a time. The operands $\vec{u}$ and $\vec{v}$ are declared as table attributes having \texttt{ARRAY} type, while the dot-product is implemented as a UDF. For reference, the existing in-database analytics alternatives treat the model as an in-memory UDA state variable~\cite{madlib,bismarck,igd-glade}. This is impossible in Big Model analytics.

\begin{figure}[htbp]
\begin{center}
\includegraphics[width=0.58\textwidth]{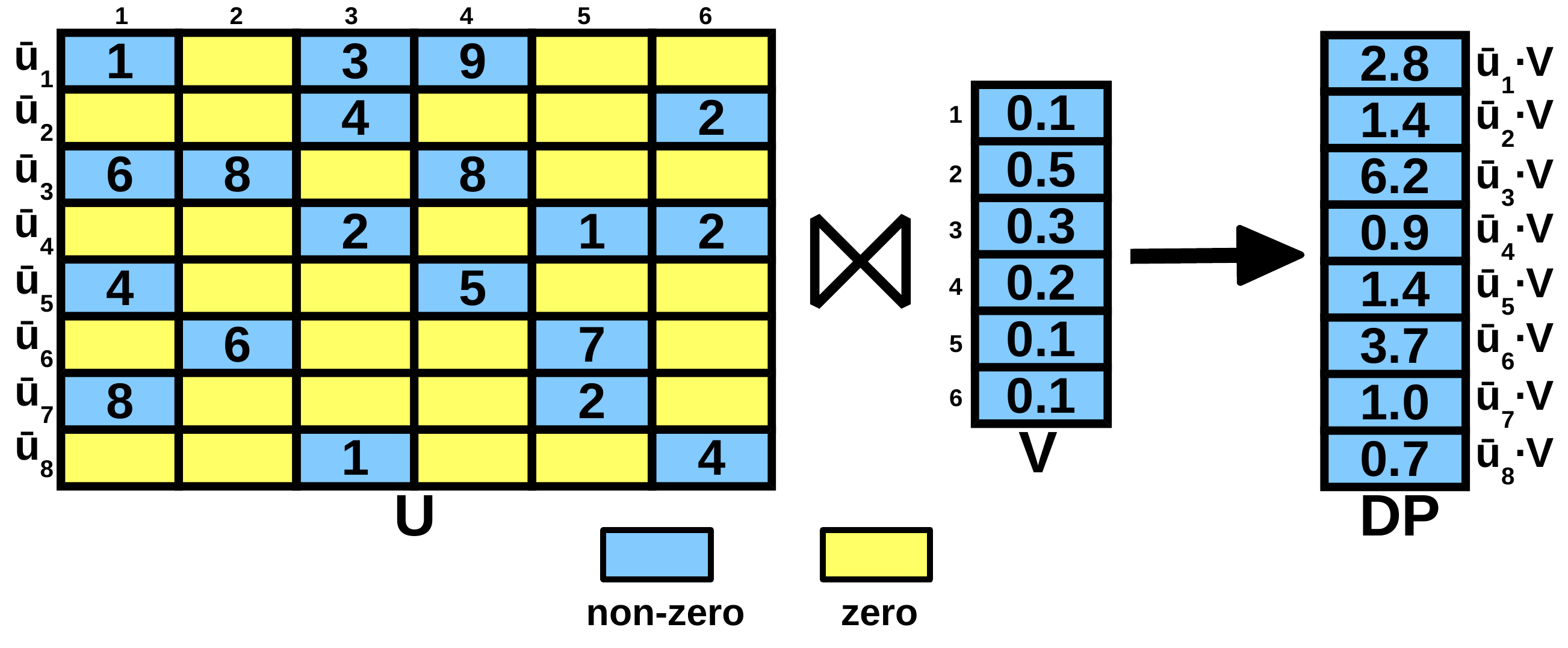}
\caption{The dot-product operator.}\label{fig:dp-operator}
\end{center}
\end{figure}

The research problem we study in this paper is \textit{how to design a database operator for Big Model dot-product}. Given a dataset consisting of sparse training examples (set of vectors \texttt{U} = $\{\vec{u}_{1}, \dots, \vec{u}_{N}\}$ in Figure~\ref{fig:dp-operator}) and a dense model (vector \texttt{V} in Figure~\ref{fig:dp-operator}), this operator computes the set of dot-products between the examples and the model (vector \texttt{DP} in Figure~\ref{fig:dp-operator}) optimally. An entry $\texttt{DP}_{i}$ corresponds to the dot-product $\vec{u}_{i}\cdot V$. Following the notation in Section~\ref{ssec:problem:grad-descent}, $\vec{u}_{i}$ corresponds to $\vec{x}_{i}$ and vector \texttt{V} corresponds to $\vec{w}$. Since the model \texttt{V} cannot fit in-memory, optimality is measured by the total number of secondary storage accesses. This is a good indicator for the execution time, given the simple computation required in dot-product. In addition to the main objective, we include a functional constraint in the design---\textit{results have to be generated in a non-blocking fashion}. As soon as the dot-product $\texttt{DP}_{i} = \vec{u}_{i}\cdot V$ corresponding to a training example $\vec{u}_{i}$ is computed, it is made available to the calling operator/application. This requirement is essential to support SGD---the gradient descent solution used the most in practice.

\textbf{Relationship between Big Model dot-product and SpMV.}
If we ignore the non-blocking functional constraint, the result vector \texttt{DP} is the product of sparse matrix \texttt{U} and vector \texttt{V}, i.e., $\texttt{DP} = \texttt{U} \cdot \texttt{V}$. This is the standard SpMV kernel which is notorious for sustaining low fractions of peak performance and for which there is a plethora of algorithms proposed in the literature. However, none of these algorithms considers the case when vector \texttt{V} does not fit in memory. They focus instead on the higher levels of the storage hierarchy, i.e., memory--cache~\cite{analytics-spmv} and global memory--texture memory in GPU~\cite{spmv-vldb}, respectively. While some of the solutions can be extended to disk-based SpMV, they are applicable only to BGD optimization which is less often used in practice. In SGD, \texttt{V} is updated after computing the dot-product with every vector $\vec{u}_{i}$, i.e., row of matrix \texttt{U}. This renders all the SpMV algorithms inapplicable to Big Model dot-product. We propose novel algorithms for this constrained version of SpMV. Since this is a special instance of SpMV, the proposed solution is applicable to the general SpMV problem.

\section{Database Dot-Product}\label{sec:baseline}

In this section, we present two database solutions to the Big Model dot-product problem. First, we introduce a pure relational solution that treats the operands \texttt{U} and \texttt{V} as standard relations. Second, we give an \texttt{ARRAY}-relation solution that represents each of the vectors $\vec{u}_{i}$ in \texttt{U} as an \texttt{ARRAY} data type---similar to existing in-database frameworks such as MADlib and Bismarck.

\subsection{Relational Dot-Product}\label{ssec:baseline:relational}

The relational solution represents \texttt{U} and \texttt{V} as relations and uses standard relational algebra operators, e.g., join, group-by, and aggregation -- and their SQL equivalent -- to compute the dot-product. In order to represent vectors $\vec{u}_{i}\in \texttt{U}$ in relational format, we create a tuple for each non-zero dimension. The attributes of the tuple include the index and the actual value in the vector. Moreover, since the components of a vector are represented as different tuples, we have to add another attribute identifying the vector, i.e., \textit{tid}. The schema for the dense vector \texttt{V} is obtained following the same procedure. There is no need for a vector/tuple identifier \textit{tid}, though, since \texttt{V} contains a single tuple, i.e., the model:
\begin{sql}
U(index INTEGER,value NUMERIC,tid INTEGER)
V(index INTEGER,value NUMERIC)
\end{sql}
Based on this representation, relation \texttt{U} corresponding to Figure~\ref{fig:dp-operator} contains 19 tuples. The tuples for $\vec{u}_{1}$ are $(1,1,1),$ $(3,3,1), (4,9,1)$. Relation \texttt{V} contains a tuple for each index. An alternative representation is a wide relation with an attribute corresponding to each dimension/index. However, tables with such a large number of attributes are not supported by any of the modern relational databases. The above representation is the standard procedure to map ordered vectors into non-ordered relations at the expense of data redundancy---the tuple identifier for sparse vectors and the index for dense vectors.

The dot-product vector \texttt{DP} is computed in the relational data representation with the following standard \textit{join group-by aggregate} SQL statement: 
\begin{sql}
SELECT U.tid, SUM(U.value*V.value)
FROM U, V
WHERE U.index=V.index
GROUP BY U.tid
\end{sql}
This relational solution supports vectors of any size. However, it breaks the original vector representation since vectors are treated as tuples. Moreover, the relational solution cannot produce the result dot-product $\texttt{DP}$ in a non-blocking manner due to the blocking nature of each operator in the query. This is not acceptable for SGD.

\subsection{Array-Relation Join}\label{ssec:baseline:array-relation}

In the \texttt{ARRAY}-relation join solution, the sparse structure of \texttt{U} is preserved by representing each vector $\vec{u}_{i}$ with two \texttt{ARRAY} attributes, one for the non-zero index and another for the value:
\begin{sql}
U(index INTEGER[],value NUMERIC[],tid INTEGER)
\end{sql}
The vector identifier \textit{tid} is still required. This representation is the database equivalent of the coordinate representation~\cite{spmv-vldb} for sparse matrices and is possible only in database servers with type extension support, e.g., PostgreSQL\footnote{\url{http://www.postgresql.org/docs/current/static/xtypes.html}}. Vector \texttt{V} is decomposed as in the relational solution. The SQL query to compute dot-product \texttt{DP} in PostgreSQL is:
\begin{sql}
SELECT U.tid, SUM(U.value[idx(U.index,V.index)]*V.value)
FROM U, V
WHERE V.index = ANY(U.index)
GROUP BY U.tid
\end{sql}
There are several differences with respect to the relational solution. The equi-join predicate in \texttt{WHERE} is replaced with an inclusion predicate \texttt{V.index IN U.index[]} testing the occurrence of a scalar value in an array. This predicate can be evaluated efficiently with an index on \texttt{V.index}. The expression in the \texttt{SUM} aggregate contains an indirection based on \texttt{V.index} which requires specialized support. While the size of the join is the same, the width of the tuples in \texttt{ARRAY}-relation join is considerably larger -- the size of the intermediate join result table is the size of U multiplied by the average number of non-zero entries across the vectors $\vec{u}_{i}$ -- since the entire \texttt{ARRAY} attributes are replicated for every result tuple. This is necessary because the arrays are used in the \texttt{SUM}. Finally, the \texttt{ARRAY}-relation join remains blocking, thus, it is still not applicable to SGD.

\section{Dot-Product Join Operator}\label{sec:dot-product}

In this section, we present the dot-product join operator for Big Model dot-product computation. Essentially, dot-product join combines the advantages of the existing database solutions. It is an in-database operator over an \texttt{ARRAY}-relation data representation that computes the result vector without blocking and without generating massive intermediate results. In this section, we present the requirements, challenges, and design decisions behind the dot-product join operator, while a thorough evaluation is given in Section~\ref{sec:experiments}.

\subsection{Requirements \& Challenges}\label{ssec:dot-product:req-chall}

The dot-product join operator takes as input the operands \texttt{U} and \texttt{V} and generates the dot-product vector \texttt{DP} (Figure~\ref{fig:dp-operator}). \texttt{U} follows the \texttt{ARRAY} representation in which the index and value are stored as sparse \texttt{ARRAY}-type attributes. Only the non-zero entries are materialized. \texttt{V} is stored in range-based partitioned relational format---the tuples \texttt{(index,value)} are partitioned into \textit{pages} with fixed size and a page contains tuples with consecutive \texttt{index} values. The overall number of pages in \texttt{V} is $p_{V}$. For example, with a page size of $2$, vector \texttt{V} in Figure~\ref{fig:dp-operator} is partitioned into $3$ pages. Page $1$ contains indexes $\{1,2\}$, page $2$ indexes $\{3,4\}$, and page $3$ indexes $\{5,6\}$. The memory budget available to the dot-product join operator for storing \texttt{V} is $M$ pages, which is smaller than $p_{V}$---this is a fundamental constraint.

Without loss of generality, we assume that each vector $\vec{u}_{i}\in U$ accesses at most $M$ pages from \texttt{V}. This allows for atomic computation of entries in vector \texttt{DP}, i.e., all the data required for the computation of the entry $\texttt{DP}_{i}$ are memory-resident at a given time instant. Moreover, each vector $\vec{u}_{i}$ has to be accessed only once in order to compute its contribution to \texttt{DP}---\texttt{U} can be processed as a stream. As a result, the cost of Big model dot-product computation is entirely determined by the number of secondary storage page accesses for \texttt{V}---the execution of Algorithm \textit{Dot-Product} is considerably faster than accessing a page from secondary storage. Thus, the main challenge faced by the dot-product join operator is \textit{minimizing the number of secondary storage page accesses} by judiciously using the memory budget $M$ for caching pages in \texttt{V}.

A vector $\vec{u}_{i}\in U$ contains several non-zero entries. Each of these entries requires a request to retrieve the value at the corresponding index in \texttt{V}. Thus, the number of requests can become a significant bottleneck even when the requested index is memory-resident. \textit{Reducing the number of requests} for entries in \texttt{V} is a secondary challenge that has to be addressed by the dot-product join operator.

\begin{ex}
In order to illustrate these challenges, we extend upon the example in Figure~\ref{fig:dp-operator}. The total number of requests to entries in \texttt{V} is $19$---$3$ for $\vec{u}_{1}$, $2$ for $\vec{u}_{2}$, and so on. The corresponding number of page accesses to \texttt{V} when we iterate over \texttt{U} is $16$---$2$ pages are accessed by each vector $\vec{u}_{i}$. This is also the number of requests when we group together requests to the same page. With a memory budget $M=2$ and LRU cache replacement policy, the number of pages accessed from secondary storage reduces to $8$. However, the number of requests remains as before, i.e., $16$.
\end{ex}

\begin{figure*}
\begin{center}
\includegraphics[width=1.0\textwidth]{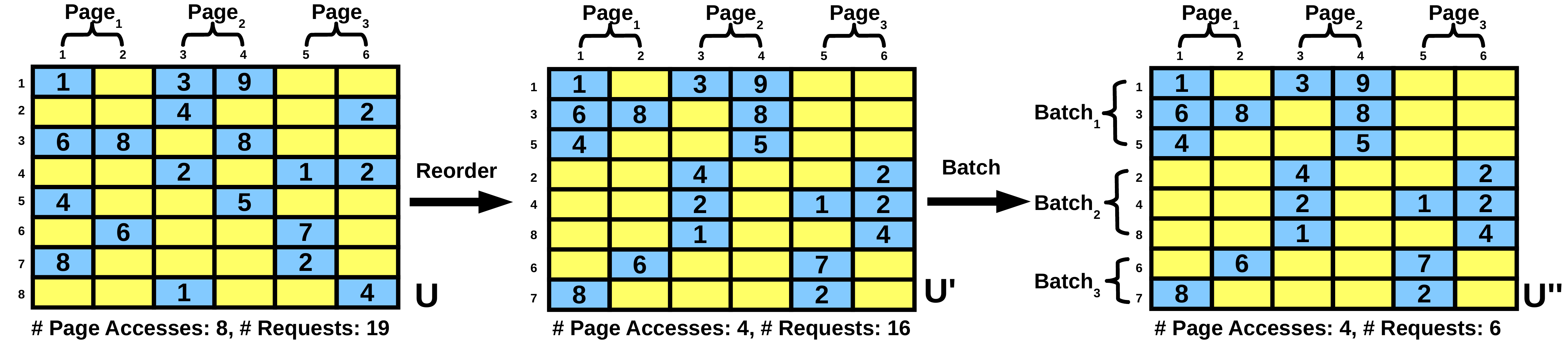}
\caption{Optimization strategies for the dot-product join operator.}\label{fig:dp-flow}
\end{center}
\end{figure*}

\subsection{High-Level Approach}\label{ssec:dot-product:high-level}

Algorithm \textit{Dot-Product Join Operator} contains a high-level presentation of the proposed operator for Big Model dot-product. Intuitively, we push the aggregation inside the \texttt{ARRAY}-relation join and apply a series of optimizations that address the challenges identified above---minimize the number of secondary storage accesses and the number of requests to entries in \texttt{V}. Given \texttt{U}, \texttt{V}, and memory budget $M$, \textit{Dot-Product Join Operator} iterates over the pages of \texttt{U} and executes a three-stage procedure at \textit{page-level}. The three stages are: optimization, batch execution, and dot-product computation. Optimization minimizes the number of secondary storage accesses. Batch execution reduces the number of requests to \texttt{V} entries. Dot-product computation is a direct call to \textit{Dot-Product} that computes an entry of the \textit{DP} result. Notice that generating intermediate join results and grouping are not required anymore since the entries in \texttt{V} corresponding to a vector $\vec{u}_{i}$ are produced in a contiguous batch. In the following, we give an overview of the optimization and batch execution stages.

\begin{algorithm}[htbp]
\caption{Dot-Product Join Operator}\label{alg:dp}
\algsetup{linenodelimiter=.}

\begin{algorithmic}[1]

\REQUIRE ~~\\
U (index INTEGER[],value NUMERIC[],tid INTEGER)\\
V (index INTEGER,value NUMERIC)\\
memory budget $M$
\ENSURE DP (tid INTEGER,product NUMERIC)

\FOR{\textbf{each} page $u_{p}\in U$}
	\item[\hspace*{.7cm}\textbf{\underline{OPTIMIZATION}}]
	\STATE Reorder vectors $\vec{u}_{i}$ to cluster similar vectors together\label{alg:dot-product:reorder}
	\STATE Group vectors $\vec{u}_{i}$ into batches $B_{j}$ that access at most $M$ pages from V\label{alg:dot-product:batching}

	\item[\hspace*{.7cm}\textbf{\underline{BATCH EXECUTION}}]
\FOR{\textbf{each} batch $B_{j}$}
	\STATE Collect pages $v_{p}\in V$ accessed by vectors $\vec{u}_{i}\in B_{j}$ into a set $v_{B_{j}} = \{ v_{p}\in V | \exists \vec{u}_{i}\in B_{j} \text{ that accesses } v_{p} \}$\label{alg:dot-product:batch-1}
	\STATE Request access to pages in $v_{B_{j}}$\label{alg:dot-product:batch-2}

	\item[\hspace*{1.0cm}\textbf{\underline{DOT-PRODUCT COMPUTATION}}]
\FOR{\textbf{each} vector $\vec{u}_{i}\in B_{j}$}\label{alg:dot-product:batch-3}
	\STATE dp $\leftarrow$ \textit{Dot-Product}($\vec{u}_{i}$,V)
	\STATE Append ($\vec{u}_{i}.\textit{tid}$, dp) to DP
\ENDFOR\label{alg:dot-product:batch-4}

\ENDFOR

\ENDFOR

\RETURN {DP}

\end{algorithmic}
\end{algorithm}

\subsubsection{Optimization}\label{ssec:dot-product:high-level:optim}

The naive adoption of the \texttt{ARRAY}-relation solution inside the dot-product join operator does not address the number of secondary storage accesses to pages in \texttt{V}. The database buffer manager is entirely responsible for handling secondary storage access. The standard policy to accomplish this is LRU---the least recently accessed page is evicted from memory. This can result in a significantly suboptimal behavior, as shown in Figure~\ref{fig:dp-flow}. \texttt{U} is the same as in Figure~\ref{fig:dp-operator} and $M=2$. As discussed before, the number of secondary storage page accesses is $8$.  Essentially, every vector $\vec{u}_{i}\in U$, except $\vec{u}_{7}$, requires a page replacement. This is because consecutive vectors access a different set of two pages. The net result is significant \textit{buffer or cache thrashing}---a page is read into memory only to be evicted at the subsequent request. One can argue that the access pattern in this example is the worst-case for LRU. The Big Model dot-product problem, however, exhibits this pattern pervasively because each dot-product computation is atomic. Allowing partial dot-product computation incurs a different set of obstacles and does not satisfy the non-blocking requirement. Thus, it is not a viable alternative for our problem.

$\texttt{U}'$ in Figure~\ref{fig:dp-flow} contains exactly the same vectors as \texttt{U}, reordered such that similar vectors -- vectors with non-zero entries co-located in the same partition/page of \texttt{V} -- are clustered together. The number of secondary storage accesses corresponding to $\texttt{U}'$ is only $4$. The reason for this considerable reduction is that a \texttt{V} page -- once read into memory -- is used by several vectors $\vec{u}_{i}$ that have non-zero indexes resident to the page. Essentially, vectors share access to the common pages.

\textbf{Reordering.}
The main task of the optimization stage is to \textit{identify the optimal reordering of vectors $\vec{u}_{i}\in U$ that minimizes the number of secondary storage accesses to \texttt{V}}---line~(\ref{alg:dot-product:reorder}) in Algorithm~\ref{alg:dp}. We prove that this task is NP-hard by a reduction from the minimum Hamiltonian path problem\footnote{\url{https://en.wikipedia.org/wiki/Hamiltonian_path}}. We propose three heuristic algorithms that find good reorderings in a fraction of the time---depending on the granularity at which the reordering is executed, e.g., a page, several pages, or a portion of a page. The first algorithm is an LSH\footnote{LSH stands for \textit{locality-sensitive hashing}.}-based extension to nearest neighbor---the well-known approximation to minimum Hamiltonian path. The second algorithm is partition-level radix sort. The third algorithm is standard k-center clustering applied at partition-level. We discuss these algorithms in Section~\ref{sec:reordering}.

Vector reordering is based on permuting a set of rows and columns of a sparse matrix in order to improve the performance of the SpMV kernel. Permuting is a common optimization technique used in sparse linear algebra. The most common implementation is the reverse Cuthill-McKee algorithm (RCM)~\cite{phi-spmv} which is widely used for minimizing the maximum distance between the non-zeros and the diagonal of the matrix. Since the RCM algorithm permutes both rows and columns, it incurs an unacceptable computational cost. To cope with this, vector reordering limits permuting only to the rows of the matrix. We show that even this simplified problem is NP-hard and requires efficient heuristics.

\textbf{Batching.}
The second task of the optimization stage is to \textit{reduce the number of page requests to the buffer manager}. Even when the requested page is in the buffer, there is a non-negligible overhead in retrieving and passing the page to the dot-product join operator. A straightforward optimization is to group together requests made to indexes co-located in the same page. By applying this strategy to the example in Figure~\ref{fig:dp-flow}, the number of requests in $\texttt{U}'$ is reduced to $16$---compared to $19$ in \texttt{U}. We take advantage of the reordering to group consecutive vectors into batches and make requests at batch-level---line~(\ref{alg:dot-product:batching}) in Algorithm~\ref{alg:dp}. This strategy is beneficial because similar vectors are grouped together by the reordering. For example, the number of page requests corresponding to $\texttt{U}''$ in Figure~\ref{fig:dp-flow} is only $6$---$2$ for each batch.

Formally, we have to identify the batches that minimize the number of page requests to \texttt{V} given a fixed ordering of the vectors in \texttt{U} and a memory budget $M$. Without the ordering constraint, this is the standard bin packing\footnote{\url{https://en.wikipedia.org/wiki/Bin_packing_problem}} problem---known to be NP-hard. Due to ordering, a simple greedy heuristic that iterates over the vectors and adds them to the current batch until the capacity constraint is not satisfied is guaranteed to achieve the optimal solution. The output is a set of batches with a variable number of vectors.

\subsubsection{Batch Execution}\label{ssec:dot-product:high-level:exec}

Given the batches $B_{j}$ extracted in the optimization stage, we compute the dot-products $\vec{u}_{i}\cdot V$ by making a single request to the buffer manager for all the pages accessed in the batch. It is guaranteed that these pages fit in memory at the same time. Notice, though, that vectors $\vec{u}_{i}$ are still processed one-by-one, thus, dot-product results are generated individually (lines~(\ref{alg:dot-product:batch-3})-(\ref{alg:dot-product:batch-4}) in Algorithm~\ref{alg:dp}). While batch execution reduces the number of requests to the buffer manager, the requests consist of a set of pages---not a single page. It is important to emphasize that decomposing the set request into a series of independent single-page requests is sub-optimal. For example, consider \texttt{U} in Figure~\ref{fig:dp-flow} to consist of $8$ batches of a single vector each. Vector $\vec{u}_{3}$ has a request for page 1 and 2, respectively. Page 2 and 3 are in memory. If we treat the 2-page request as two single-page requests, page 2 is evicted to release space for page 1, only to be immediately read back into memory. If the request is considered as a unit, page 3 is evicted.

In order to support set requests, the functionality of the buffer manager has to be enhanced. Instead of directly applying the replacement policy, e.g., LRU, there are two stages in handling a page set request. First, the pages requested that are memory-resident have to be pinned down such that the replacement policy does not consider them. Second, the remaining requested pages and the non-pinned buffered pages are passed to the standard replacement policy. \textit{Dot-Product Join Operator} includes this functionality in lines (\ref{alg:dot-product:batch-1})-(\ref{alg:dot-product:batch-2}), before the dot-product computation.

\section{Vector Reordering}\label{sec:reordering}

In this section, we discuss strategies for reordering vectors $\vec{u}_{i}\in U$ that minimize the number of secondary storage accesses. First, we prove that finding the optimal reordering is NP-hard. Then, we introduce three scalable heuristics that cluster similar vectors.

\eat{
We introduce reordering technique to optimize the order of tuple in \texttt{U} such that tuples with similar super-lookups are placed ``closely'', if not right next to each other. The ordering is done at two level as well---tuple level and batch level. Reordering is first done at the tuple level to achieve maximum batch size and then done at the batch level to maximize the page sharing across batches. 

To show how the order matters, we give an example matrix U of figure \ref{fig:dp-flow}. Assume that the cache can fit 2 pages at most so that a batch can hold a maximum of 2 super-lookups. As the order of the tuples are given in matrix U, there are only two tuples can be put in a shared batch, i.e. $tuple_{6}$ and $tuple_{7}$. All the other tuples have to be a batch by their own since any merging any consecutive tuples other than $tuple_{6}$ and $tuple_{7}$ would result in a batch with three super-lookups. Thus, even after batching is applied within tuples, there are 14 super-lookups overall ($tuple_{6}$ and $tuple_{7}$ share two super-lookups and each of the rest have two super-lookups). However, if a better order is given as matrix U', tuples can be batched as in U'' and the resulting number of super-lookups is only 6. Same reordering is applied at the batch level to maximize the cache reuse between batches, though in case of matrix U'', there is no difference between different batch orders and there will always be two page misses.
}

\subsection{Optimal Reordering is NP-hard}\label{ssec:reordering:np-hard}

We formalize the reordering problem as follows. Assume there are $N$ $d$-dimensional vectors $\{\vec{u}_{1}, \dots, \vec{u}_{N}\}$. Each vector $\vec{u}_{i}$ contains $r_{i}$ page requests grouped into a set $v_{p}^{i}=\{p_{1}^{i}, \dots, p_{r_{i}}^{i}\}$. Given two consecutive vectors $\vec{u}_{i}$, $\vec{u}_{i+1}$ and their corresponding page request sets $v_{p}^{i}$, $v_{p}^{i+1}$, the pages in the set difference $v_{p}^{i+1}\setminus v_{p}^{i}$ potentially require secondary storage access. Since we assume that all the pages in any set $v_{p}^{i}$ fit in memory, it is guaranteed that the pages in the set intersection $v_{p}^{i+1}\cap v_{p}^{i}$ are memory resident. Pages in $v_{p}^{i+1}$ that are not in $v_{p}^{i}$ can be in memory if they have been accessed before and not evicted. Let us denote the set difference between any two vectors $\vec{u}_{i}$ and $\vec{u}_{j}$ as $C^{i,j} = v_{p}^{i}\setminus v_{p}^{j}$. The cardinality of this set $\left|C^{i,j}\right|$ gives the number of potential secondary storage accesses. The goal of reordering is to identify the vector sequence $\{\vec{u}_{i_{1}}, \dots, \vec{u}_{i_{N}}\}$ that minimizes the cumulative cardinality $\left|C^{i_{j+1},i_{j}}\right|$ of the set difference between all the pairs of consecutive vectors:
\begin{equation}\label{eq:optim-order}
\textit{min}_{\{\vec{u}_{i_{1}}, \dots, \vec{u}_{i_{N}}\}} \sum_{j=1}^{N-1} \left|C^{i_{j+1},i_{j}}\right|
\end{equation}

\begin{thm}\label{thm:reordering-np-hard}
Finding the vector sequence $\{\vec{u}_{i_{1}}, \dots, \vec{u}_{i_{N}}\}$ that minimizes the cumulative cardinality $\sum_{j=1}^{N-1} \left|C^{i_{j+1},i_{j}}\right|$ is NP-hard.
\end{thm}
\begin{proof}
We provide a reduction of the minimum Hamiltonian path in a weighted directed graph -- a known NP-complete problem~\cite{hamilton-np-complete} -- to the reordering problem. Given a weighted directed graph -- the edges are directional and have a weight attached -- the minimum Hamiltonian path traverses all the vertexes once and the resulting path has minimum weight.
We define a complete directed graph $G=(V,E)$ with $N$ vertexes and $N(N-1)$ edges. A vertex corresponds to each vector $\vec{u}_{i}$. For any pair of vertexes, i.e., vectors $\vec{u}_{i}$, $\vec{u}_{j}$, we introduce two edges $(\vec{u}_{i}$, $\vec{u}_{j})$ and $(\vec{u}_{j}$, $\vec{u}_{i})$, respectively. The weight of an edge $(\vec{u}_{i}$, $\vec{u}_{j})$ is given by the cardinality $\left|C^{j,i}\right| = \left|v_{p}^{j}\setminus v_{p}^{i}\right|$. Notice that set difference is not commutative. This is the reason for having directed edges. A Hamiltonian path in graph $G$ corresponds to a reordering because all the vectors are considered only once and all the orders are possible---there is an edge between any two vertexes. Since the weight of an edge is defined as the cardinality $\left|C^{i,j}\right|$, the weight of the minimum Hamiltonian path corresponds to the minimum cumulative cardinality in Eq.~(\ref{eq:optim-order}).
\end{proof}

\eat{
\begin{proof}
We prove the theorem by reduction from minimum Hamiltonian path problem on a complete graph. In a complete weighted directed graph $G$ with $N$ nodes, denote the edge weight between $node_{i}$ and $node_{j}$ as $e_{ij}$ and $e_{ji}$. A set of $N$ tuples $B$=\{$tuples_{1}$, $tuples_{2}$, ..., $tuples_{N}$\} can be constructed such that the difference between any two tuples $C_{ij}$ equals to $e_{ji}$. Thus, if there exist a polynomial-time algorithm can generate an optimal sequence that minimize the cumulative consecutive cost of set tuple $B$. We can use this sequence to obtain a Hamiltonian path from graph $G$ with minimum cost. Since finding a minimum Hamiltonian path in complete graph is NP-hard, there is no polynomial-time algorithm for our order finding problem unless $P=NP$.
\end{proof}
}

Several polynomial-time heuristic algorithms for computing an approximation to the minimum Hamiltonian path over a complete directed graph exist in the literature~\cite{approx-Hamiltonian-path}. Nearest neighbor is a standard greedy algorithm that chooses the closest non-visited vertex as the next vertex. It has computational complexity of $\mathcal{O}(N^{2})$ and non-constant approximation factor $\log(N)$. The straightforward application of the nearest neighbor algorithm to our problem is too expensive because we have to materialize a complete directed graph. This takes significant time even for values of $N$ in the order of thousands---not to mention the required space. If we generate the graph on-the-fly, the weight to all the non-visited vertexes has to be computed. While this may seem more efficient, this still has an overall computational complexity of $\mathcal{O}(N^{2}d)$. Due to these limitations of the nearest neighbor algorithm for complete graphs, we explore more efficient heuristics that avoid considering all the possible pairs of vertexes.

\subsection{LSH-based Nearest Neighbor}\label{ssec:reordering:lsh}

Locality-sensitive hashing (LSH)~\cite{lsh} is an efficient method to identify similar objects represented as high-dimensional sparse vectors. Similarity between vectors is defined as their Jaccard coefficient, i.e., $J(\vec{u}_{i},\vec{u}_{j}) = \frac {\left|v_{p}^{i}\cap v_{p}^{j}\right|} {\left|v_{p}^{i}\cup v_{p}^{j}\right|}$. The main idea is to build a hash index that groups similar vectors in the same bucket. Given an input vector, the most similar vector is found by identifying the vector with the maximum Jaccard coefficient between the vectors co-located in the same hash bucket. Essentially, the search space is pruned from all the vectors in the set to the vectors in the hash bucket---typically considerably fewer. The complexity of LSH consists in finding a hash function that preserves the Jaccard coefficient in mapping vectors to buckets---the probability of having two vectors in the same bucket approximates their Jaccard coefficient. Minwise hash functions~\cite{minwise-hash} satisfy this property on expectation. Given a set, a minwise hash function generates any permutation of its elements with equal probability. While such functions are hard to define, they can be approximated with a universal hash function~\cite{universal-hash} that has a very large domain, e.g., $2^{64}$. In order to increase the accuracy of correctly estimating the Jaccard coefficient, $m$ such minwise hash functions are applied to a vector. Their output corresponds to the signature of the vector which gives the bucket where the vector is hashed to. The value of $m$ is an important parameter controlling the number of vectors which are exhaustively compared to the input query vector. The larger $m$ is, the fewer vectors end-up in the same bucket. On the opposite, if $m=1$, all the vectors that share a common value can be hashed to the same bucket. Given a value for $m$, banding is a method that controls the degree of tolerated similarity. The $m$-dimensional signature is divided into $b$ $\left\lfloor m/b \right\rfloor$-dimensional bands and a hash table is built independently for each of them. The input vector is compared with all the co-located vectors of at least one hash table. Banding decreases the Jaccard coefficient threshold acceptable for similarity, while increasing the probability that all the vectors that have a higher coefficient than the threshold are found.

\begin{algorithm}[htbp]
\caption{LSH Reordering}\label{alg:lsh}
\algsetup{linenodelimiter=.}

\begin{algorithmic}[1]

\REQUIRE ~~\\
Set of vectors $\{\vec{u}_{1}, \dots, \vec{u}_{N}\}$ with page requests\\
$m$ minwise hash functions grouped into $b$ bands
\ENSURE Reordered set of input vectors $\{\vec{u}_{i_{1}}, \dots, \vec{u}_{i_{N}}\}$

\item[\underline{Compute LSH tables}]
\FOR{\textbf{each} vector $\vec{u}_{i}$}
	\STATE Compute $m$-dimensional signature $(s_{1}, \dots, s_{m})$ based on minwise hash functions and group into $b$ bands $\textit{band}_{k} = (s_{(k-1)\cdot \left\lfloor \frac{m}{b} \right\rfloor + 1}, \dots, s_{k\cdot \left\lfloor \frac{m}{b} \right\rfloor})$
	\STATE Insert vector $\vec{u}_{i}$ into hash table $\textit{Hash}_{k}$, $1\leq k \leq b$, using minwise hash function of $\textit{band}_{k}$
\ENDFOR

\item[\underline{LSH-based nearest neighbor search}]
\STATE Initialize $\vec{u}_{i_{1}}$ with a random vector $\vec{u}_{i}$
\FOR{$j=1$ \textbf{to} $N-1$}
	\STATE Let $X_{k}$ be the set of vectors co-located in the same bucket with $\vec{u}_{i_{j}}$ in hash table $\textit{Hash}_{k}$ and not selected
	\STATE $X \leftarrow X_{1} \cup \dots \cup X_{b}$
	\STATE Let $\vec{u}_{i_{j+1}}$ be the vector in $X$ with the minimum set difference cardinality $\left|C^{i_{j+1},i_{j}}\right|$ to the current vector $\vec{u}_{i_{j}}$
\ENDFOR

\end{algorithmic}
\end{algorithm}

\textit{Compute LSH tables} section in Algorithm~\ref{alg:lsh} summarizes the construction of the LSH index for the vector reordering problem. Although we do not measure similarity using the Jaccard coefficient, there is a strong correlation between set difference cardinality and the Jaccard coefficient. Intuitively, the higher the Jaccard coefficient, the larger the intersection between two sets relative to their union. This translates into small set difference cardinality. Since the goal of reordering is to cluster vectors with small differences, LSH places them into the same bucket with high probability.

We compute the output vector reordering by executing the nearest neighbor heuristic over the LSH index (section \textit{LSH-based nearest neighbor search} in Algorithm~\ref{alg:lsh}). We believe this is a novel application of the LSH technique, typically used for point queries. The algorithm starts with a random vector. The next vector is selected from the vectors co-located in the same bucket across at least one other band of the LSH index. The process is repeated until all the vectors are selected. Bands play a very important role in reordering because they allow for smooth bucket transition. This is not possible with a single band since there is no strict ordering between buckets.  In this situation, choosing the next bucket involves inspecting all the other buckets of the hash table.

In general, LSH reduces significantly the number of vector pairs for which the exact set difference has to be computed---only $\mathcal{O}(N)$ vector pairs are considered, i.e., a constant number for each vector. Thus, the overall complexity of \textit{LSH Reordering} -- $\mathcal{O}(Ndm)$ -- is dominated by the construction of the LSH index.

\begin{ex}
We illustrate how LSH reordering works for the set of vectors $U$ depicted in Figure~\ref{fig:dp-flow}. To facilitate understanding, we set $m=2$ and $b=2$, i.e., two LSH indexes with 1-D signatures. Since $m/b=1$, there are many conflicts in each bucket of the two bands. Even though this does not reduce dramatically the number of vector pairs that require full comparison, it shows how bucket transition works. Let the two minwise hash functions generate the following permutations: $\{2, 3, 1\}$ and $\{3, 1, 2\}$, respectively. Remember that the reordering is done at page level. The LSH index for the first band has two buckets, for key $2: \{\vec{u}_{1},\vec{u}_{2},\vec{u}_{3},\vec{u}_{4},\vec{u}_{5},\vec{u}_{8}\}$ and key $3: \{\vec{u}_{6},\vec{u}_{7}\}$. The LSH index for the second band also has two buckets, for key $1: \{\vec{u}_{1},\vec{u}_{3},\vec{u}_{5}\}$ and key $3: \{\vec{u}_{2},\vec{u}_{4},\vec{u}_{6},\vec{u}_{7},\vec{u}_{8}\}$. Let $\vec{u}_{1}$ be the random vector we start the  nearest neighbor search from. The vectors considered at the first step of the algorithm are $\{\vec{u}_{2},\vec{u}_{3},\vec{u}_{4},\vec{u}_{5},\vec{u}_{8}\}$. All of them are contained in bucket with key $2$ of the first band. Since $\vec{u}_{3}$ and $\vec{u}_{5}$ have set difference 0 to $\vec{u}_{1}$, one of them is selected as $\vec{u}_{i_{2}}$ and the other as $\vec{u}_{i_{3}}$. At this moment, the bucket with key $1$ from the second band is exhausted and one of $\vec{u}_{2}$, $\vec{u}_{4}$, and $\vec{u}_{8}$ is selected. Independent of which one is selected, bucket transition occurs since the new vectors $\vec{u}_{6},\vec{u}_{7}$ are co-located in bucket with key $3$ of the second band. By following the algorithm to termination, $U'$ in Figure~\ref{fig:dp-flow} is a possible solution.
\end{ex}

\subsection{Radix Sort}\label{ssec:reordering:sorting}

Sorting is a natural approach to generate a reordering of a set as long as a strict order relationship can be defined between any two elements of the set. The Jaccard coefficient gives an order only with respect to a fixed reference, i.e., given a reference vector we can quantify which of any other two vectors is smaller based on their Jaccard coefficient with the reference. However, this is not sufficient to generate a complete reordering.

It is possible to imagine a multitude of ordering strategies for a set of sparse $d$-dimensional vectors. The simplest solution is to consider the dimensions in some arbitrary order, e.g., from left to right. Vector $\vec{u}_{1}$ is smaller than $\vec{u}_{2}$, i.e., $\vec{u}_{1} < \vec{u}_{2}$, for $U$ in Figure~\ref{fig:dp-flow} since it has a non-zero entry at index $1$, while the first non-zero entry in $\vec{u}_{2}$ is at index $3$. In order to cluster similar vectors together, a better ordering is required. Our strategy is to sort the dimensions according to the frequency at which they appear in the set of sparse vectors and compare vectors dimension-wise based on this order. This is exactly how radix sort~\cite{algorithms-intro} works, albeit without reordering the dimensions in descending order of their frequency. A similar idea is proposed for SpMV on GPU in~\cite{spmv-vldb}. Algorithm \textit{Radix Sort Reordering} depicts the entire process. The two stages -- frequency computation and radix sort -- are clearly delimited. Since their complexity is $\mathcal{O}(Nd)$, the overall is $\mathcal{O}(Nd)$.

\begin{algorithm}[htbp]
\caption{Radix Sort Reordering}\label{alg:radix}
\algsetup{linenodelimiter=.}

\begin{algorithmic}[1]

\REQUIRE Set of vectors $\{\vec{u}_{1}, \dots, \vec{u}_{N}\}$ with page requests
\ENSURE Reordered set of input vectors $\{\vec{u}_{i_{1}}, \dots, \vec{u}_{i_{N}}\}$

\item[\underline{Page request frequency computation}]
\STATE Compute page request frequency across vectors $\{\vec{u}_{1}, \dots, \vec{u}_{N}\}$
\FOR{\textbf{each} vector $\vec{u}_{i}$}
	\STATE Represent $\vec{u}_{i}$ by a bitset of 0's and 1's where a 1 at index $k$ corresponds to the vector requesting page $k$
	\STATE Reorder the bitset in decreasing order of the page request frequency, i.e., index 1 corresponds to the most frequent page
\ENDFOR

\item[\underline{Radix sort}]
\STATE Apply radix sort to the set of bitsets
\STATE Let $\vec{u}_{i_{j}}$ be the vector corresponding to the bitset at position $j$ in the sorted order
\end{algorithmic}
\end{algorithm}

\begin{ex}
We illustrate how radix sort reordering works for the set of vectors $U$ in Figure~\ref{fig:dp-flow}. The algorithm operates at page level. Since page $2$ is the most frequent -- it appears in $6$ vectors -- it is the first considered. Two partitions are generated: $p_{0} = \{\vec{u}_{1},\vec{u}_{3},\vec{u}_{5},\vec{u}_{2},\vec{u}_{4},\vec{u}_{8}\}$ accesses page $2$; and $p_{1} = \{\vec{u}_{6},\vec{u}_{7}\}$ does not. This is exactly $U'$ in Figure~\ref{fig:dp-flow}. The frequency of page $1$ and page $3$ is $5$, thus, any of them can be selected in the second iteration. Assume that we brake the ties using the original index and we select page $1$. Each of the previous two partitions is further split into two. For example, $p_{00} = \{\vec{u}_{1},\vec{u}_{3},\vec{u}_{5}\}$ and $p_{01} = \{\vec{u}_{2},\vec{u}_{4},\vec{u}_{8}\}$. The important thing to remark is that vectors accessing page $1$ split in the first iteration cannot be grouped together anymore. If we follow the algorithm to completion, $U'$ in Figure~\ref{fig:dp-flow} is generated.
\end{ex}

The intuition behind radix sort reordering is that -- by considering pages in decreasing order of their frequency -- the most requested page is accessed from secondary storage exactly once; the second most accessed page at most twice; and so on. This holds true because all the pages accessed by a vector fit together in memory. It is important to notice that -- although the number of accesses increases -- the request frequency decreases for pages considered at later iterations. This guarantees that the maximum number of accesses to a page is bounded by $\textit{min}\left\{2^{\textit{rank}},\textit{freq}\right\}$, where \textit{rank} is the rank of the page and \textit{freq} is the access frequency. Essentially, radix sort reordering is a dimension-wise greedy heuristic.

\subsection{K-Center Clustering}\label{ssec:reordering:clustering}

Since the goal of reordering is to cluster similar vectors together, we include a standard k-center clustering~\cite{approximation-algorithms} algorithm as a reference. The main challenge is that -- similar to the Jaccard coefficient -- clustering does not impose a strict ordering between two vectors---only with respect to the common center. This is also true for centers. The stopping criterion for our hierarchical k-center clustering is when all the clusters have page requests that fit entirely in memory. As long as this does not hold, we invoke the algorithm recursively for the clusters that do not satisfy this requirement. The resulting clusters are ordered as follows. We start with a random cluster and select as the next cluster the one with the center having the minimum set difference cardinality. Notice that reordering the vectors inside a cluster is not necessary since they all fit in memory. This procedure is depicted in Algorithm~\ref{alg:k-center}. It has complexity $\mathcal{O}(Ndk)$, where $k$ is the resulting number of clusters, i.e., centers.

\begin{algorithm}[htbp]
\caption{K-Center Reordering}\label{alg:k-center}
\algsetup{linenodelimiter=.}

\begin{algorithmic}[1]

\REQUIRE Set of vectors $\{\vec{u}_{1}, \dots, \vec{u}_{N}\}$ with page requests
\ENSURE Reordered set of input vectors $\{\vec{u}_{i_{1}}, \dots, \vec{u}_{i_{N}}\}$

\STATE Initialize first set of $k$ centers with random vectors $\vec{u}_{i}$
\STATE Assign each vector to the center having the minimum set difference cardinality
\STATE Let $X_{l}$ be the set of vectors assigned to center $l$, $1 \leq l \leq k$
\STATE Call \textit{K-Center Reordering} recursively for the sets $X_{l}$ with requests that do not fit in memory

\STATE Reorder centers and their corresponding vectors

\end{algorithmic}
\end{algorithm}
\subsection{Discussion}\label{ssec:reordering:discussion}

We propose three heuristic algorithms for the vector reordering problem. LSH and k-center cluster similar vectors together. LSH uses the Jaccard coefficient to hash similar vectors to the same bucket of a hash table and then executes nearest neighbor search starting from a random vector. K-center clustering partitions the vectors recursively based on an increasing set of centers. Both methods are limited by partial ordering between two vectors and incur overhead to define a strict ordering. Moreover, they are randomized algorithms sensitive to the initialization and a handful of parameters. Radix sort imposes a strict ordering at the expense of not considering the entire vector in clustering. We alleviate this problem by sorting the dimensions based on their access frequency. This bounds the total number of secondary storage accesses. In the experimental evaluation (Section~\ref{sec:experiments}), we compare these algorithms thoroughly in terms of execution time and reordering quality.

\section{Gradient Descent Integration}\label{sec:dp-gd}

\begin{minipage}{.35\textwidth}
In this section, we show how to integrate the dot-product join operator in gradient descent optimization for Big Model analytics. We discuss the benefits of the operator approach compared to the relational and \texttt{ARRAY}-relation solutions presented in Section~\ref{sec:baseline}.

Figure~\ref{fig:sgd-integration} depicts the gradient computation required in gradient descent optimization. Vector dot-product is only a subroutine in this process. As we move from the relational solution to the proposed dot-product join operator, the query plan becomes considerably simpler. The relational solution consists of two parts. In the first part, the dot-product corresponding to a vector $\vec{u}_{i}$ is computed. Since vector components are represented as independent tuples with a common identifier, this requires a group-by on \textit{tid}. However, this results in the loss of the vector components, required for gradient computation (see Eq.~(\ref{eq:lr-gradient})). Thus, a second join group-by is necessary in order to compute each component of the gradient.
\end{minipage}\hfill
\begin{minipage}{.6\textwidth}
\begin{figure}[H]
\begin{center}
\includegraphics[width=\textwidth]{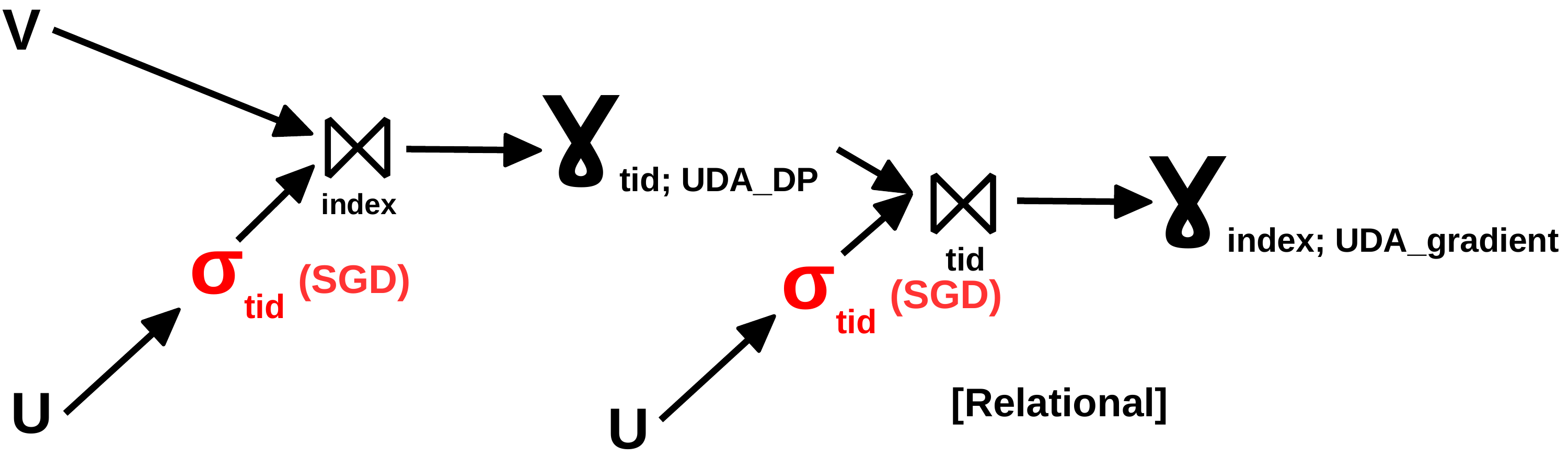}\\
\includegraphics[width=\textwidth]{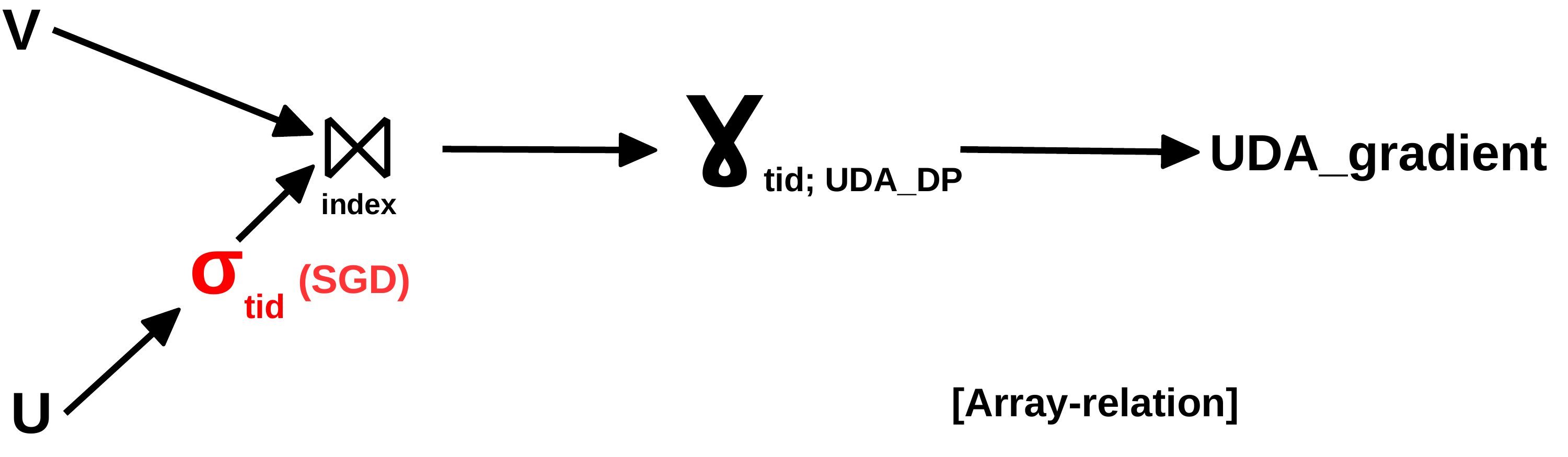}\\
\includegraphics[width=\textwidth]{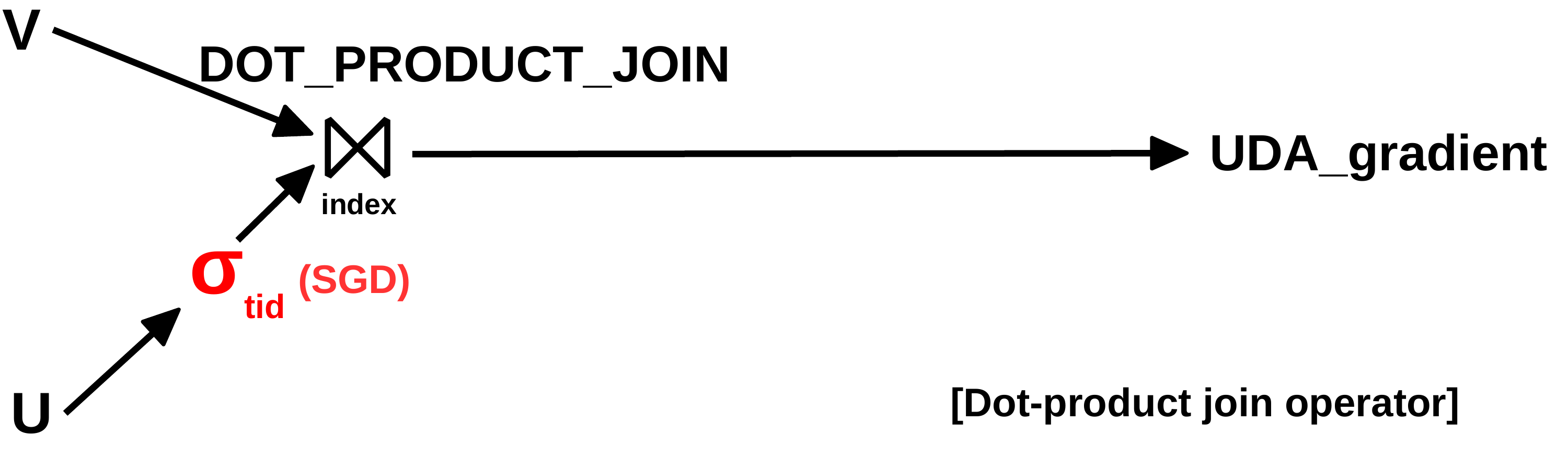}\\
\caption{Gradient descent integration.}\label{fig:sgd-integration}
\end{center}
\end{figure}
\end{minipage}\hfill
As we show in the experimental evaluation, this turns out to be very inefficient. The \texttt{ARRAY}-relation solution is able to discard the second join group-by because it groups vector components as an array attribute of a tuple.
 
The proposed dot-product join operator goes one step further and pushes the group-by aggregation inside the join. While this idea has been introduced in~\cite{sgd-over-join} for BGD over normalized example data, the dot-product join operator considers joins between examples and the model and works for BGD and SGD alike---SGD requires only an additional selection. In~\cite{sgd-over-join}, the model $V$ is small enough to fit in the state of the UDA. The main benefit of pushing the group-by aggregation inside the join is that the temporary join result is not materialized---in memory or on secondary storage. The savings in storage can be several orders of magnitude, e.g., with an average of $1000$ non-zero indexes per vector $\vec{u}_{i}$, the temporary storage -- if materialized -- is 3 orders of magnitude the size of $U$. By discarding the blocking group-by on \textit{tid}, the overall gradient computation becomes non-blocking since dot-product join is non-blocking.

\textbf{SGD considerations.}
In order to achieve faster convergence, SGD requires random example traversals. Since dot-product join reorders the examples in order to cluster similar examples together, we expect this to have a negative effect on convergence. However, the reordering in dot-product join is only local---at page-level. Thus, a simple strategy to eliminate the effect of reordering completely is to estimate the gradient at page-level---the number of steps in an iteration is equal to the number of pages. Any intermediate scheme that trades-off convergence speed with secondary storage accesses can be imagined. To maximize convergence, the data traversal orders across iterations have to be also random. This is easily achieved with LSH and K-center reordering---two randomized algorithms. A simple solution for Radix is to randomize the order of pages across iterations.

\begin{figure}[htbp]
\begin{center}
\includegraphics[width=.6\textwidth]{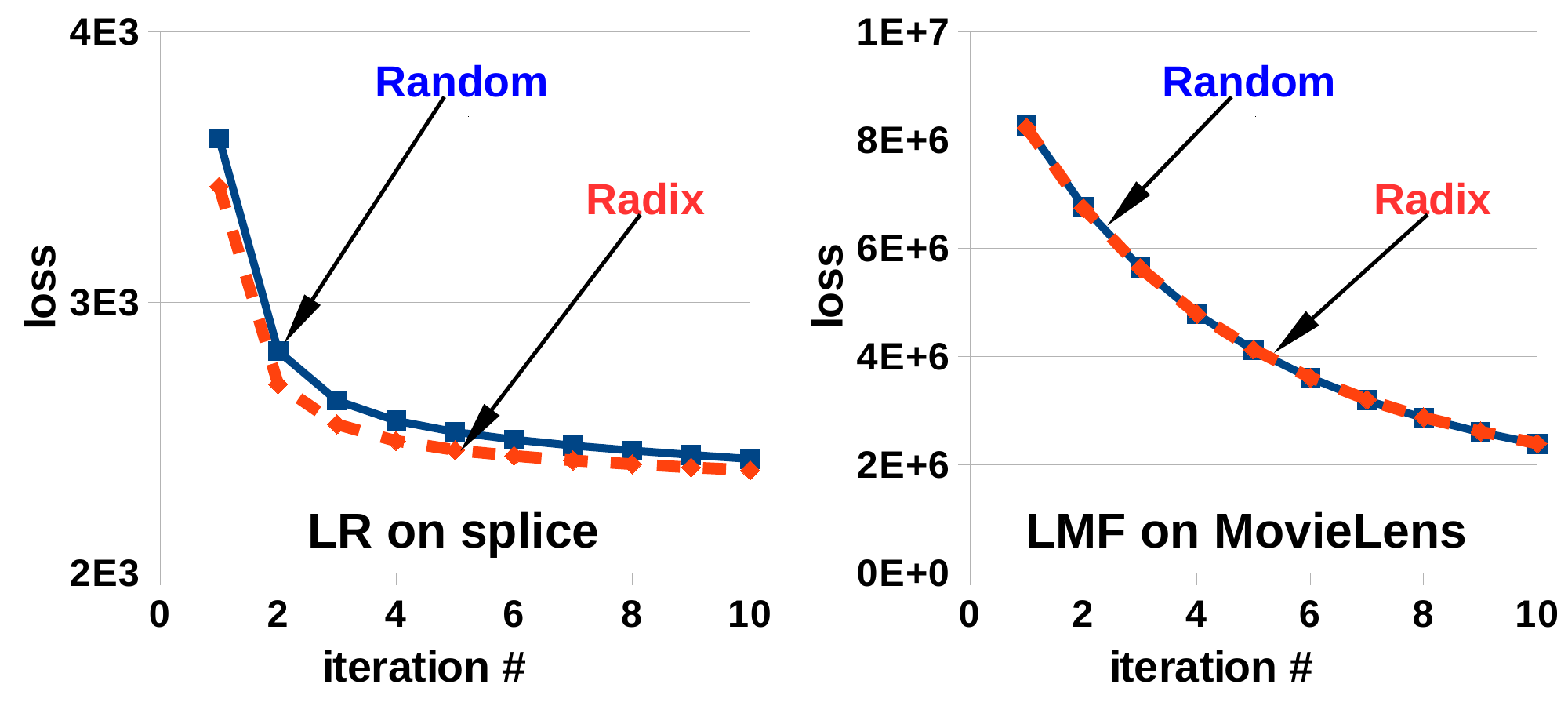}
\caption{Effect of Radix reordering on SGD convergence.}\label{fig:sgd-convergence}
\end{center}
\end{figure}

We provide experimental evidence that quantifies the impact of Radix reordering -- the strictest solution -- on convergence. Figure~\ref{fig:sgd-convergence} depicts the behavior of the loss function when 10 SGD iterations are executed for LR and LMF models over two real datasets (see Section~\ref{sec:experiments} for details). Radix reordering does not degrade the convergence speed compared to a random data traversal---the model is updated after each example. Moreover, for the LR model, Radix reordering actually improves convergence. These results are in line with those presented in~\cite{bismarck}, where only complete data sorting on the label has a negative impact on convergence.

\section{Experimental Evaluation}\label{sec:experiments}

In the section, we first evaluate the effectiveness and efficiency of the three reordering heuristics. Then, we apply the reordering heuristics to dot-product join and measure the effect of reordering and batching under different resource constraints. Finally, we measure the end-to-end dot-product and gradient descent execution time as a function of the amount of memory available in the system. We also compare dot-product join with the baseline alternatives introduced in Section~\ref{sec:baseline} across several synthetic and real datasets. Specifically, the experiments we design are targeted to answer the following questions:
\begin{compactitem}
\item How effective are the reordering heuristics and which one should be used under what circumstance?
\item How do reordering and batching affect the runtime and what is their overhead?
\item What is the sensitivity of the dot-product join operator with respect to the available memory?
\item How does dot-product join compare with other solutions?
\item What is the contribution of dot-product join within the overall Big Model gradient descent optimization?
\end{compactitem}

\subsection{Setup}\label{sec:experiments:setup}

\textbf{Implementation.}
We implement dot-product join as a new array-relation operator in GLADE~\cite{glade:sigmod}---a state-of-the-art parallel data processing system that executes analytics tasks expressed with the UDA interface. GLADE has native support for the \texttt{ARRAY} data type. Dot-product join is implemented as an optimized index join operator. It iterates over the pages in \texttt{U}. For each page, it applies the reordering and batching optimizations and then probes the entries in \texttt{V} at batch granularity. The dot-product corresponding to a vector is generated in a single pass over the vector and pipelined into the gradient UDA.

\textbf{System.}
We execute the experiments on a standard server running Ubuntu 14.04 SMP $64$-bit with Linux kernel 3.13.0-43. The server has 2 AMD Opteron 6128 series 8-core processors -- 16 cores -- 28 GB of memory, and 1 TB 7200 RPM SAS hard-drive. Each processor has 12 MB L3 cache, while each core has 128 KB L1 and 512 KB L2 local caches. The average disk bandwidth is 120 MB/s.

\textbf{Methodology.}
We perform all experiments at least 3 times and report the average value as the result. In the case of page-level results, we execute the experiments over the entire dataset -- all the pages -- and report the average value computed across the pages. We always enforce data to be read from disk in the first iteration by cleaning the file system buffers before execution. Memory constraints are enforced by limiting the batch size, i.e., the number of vectors $\vec{u}_{i}$ that can be grouped together after reordering.

\begin{table}[htbp]
  \begin{center}
    \begin{tabular}{l||r|r|r|r}

	\textbf{Dataset} & \textbf{\# Dims} & \textbf{\# Examples} &\textbf{Size}& \textbf{Model}\\

	\hline
	
	\texttt{uniform} & 1B & 80K & 4.2 GB & 8 GB\\
	
	\texttt{skewed} & 1B & 1M & 4.5 GB & 8 GB\\

	\texttt{matrix} & 10M x 10K & 300M & 4.5 GB & 80 GB\\

	\hline

	\texttt{splice} & 13M & 500K & 30 GB & 100 MB\\

	\texttt{MovieLens} & 6K x 4K & 1M & 24 MB & 80 MB\\
    \end{tabular}
  \end{center}

\caption{Datasets used in the experiments.}\label{tbl:datasets}
\end{table}

\textbf{Datasets and tasks.}
We run experiments over five datasets---three synthetic and two real. Table~\ref{tbl:datasets} shows their characteristics. \texttt{uniform} and \texttt{skewed} contain sparse vectors with dimensionality 1 billion having non-zero entries at random indexes. For \texttt{uniform}, the non-zero indexes are chosen with uniform probability over the domain. On average, there are $3000$ non-zero entries for each example vector. In the case of \texttt{skewed}, the frequency of non-zero indexes is extracted from a zipf distribution with coefficient $1.0$. The index and the vectors are randomly generated. On average, there are only $300$ non-zero entries for each vector. However, the number of example vectors is much larger---at least as large as the highest frequency. The size of the model for both datasets is $8$ GB since we store the model in dense format. While not every index is accessed, each page of the model is accessed. \texttt{matrix} is generated following the same process, with the additional constraint that there is at least a non-zero entry for each index---if there is no rating for a movie/song, then it can be removed from the data altogether. The properties of \texttt{matrix} follow closely the Spotify example given in the introduction. \texttt{splice}~\cite{vowpal-wabbit} is a real massive dataset for distributed gradient descent optimization, $3.2$ TB in full size. We extract a $1\%$ sample for our single-disk experiments. However, the dimensionality of the model is preserved. We notice that \texttt{splice} is, in fact, a uniform dataset. \texttt{MovieLens}~\cite{bismarck} is a small real dataset -- both in terms of number of examples and dimensions -- that we include mainly to study the impact of reordering on convergence speed. It is important to emphasize that model size, i.e., dimensionality, is the main performance driver. The number of examples, i.e., size, has a linear impact on execution time. We evaluate the vector dot-product independently as well as part of SGD for LR and LMF models. We execute LR over \texttt{uniform}, \texttt{skewed}, and \texttt{splice}; LMF with rank $1000$ over \texttt{matrix} and \texttt{MovieLens}. We present only the most relevant results, although we execute experiments for all combinations of datasets, tasks, and parameter configurations.

\begin{figure*}[htbp]
\begin{center}
\subfloat[]{\includegraphics[width=0.48\textwidth]{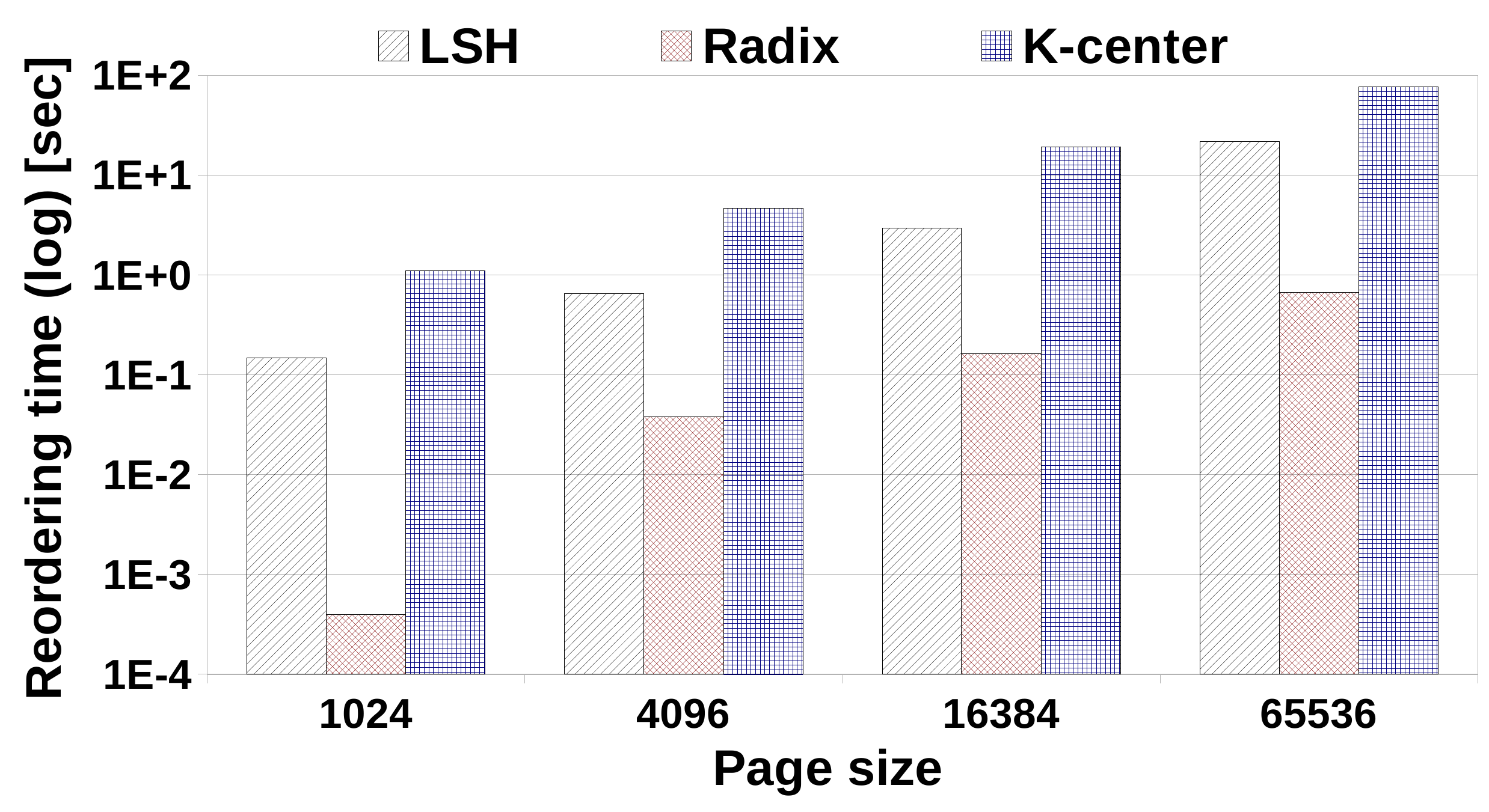}\label{fig:opt-time-comp}}
\subfloat[]{\includegraphics[width=0.48\textwidth]{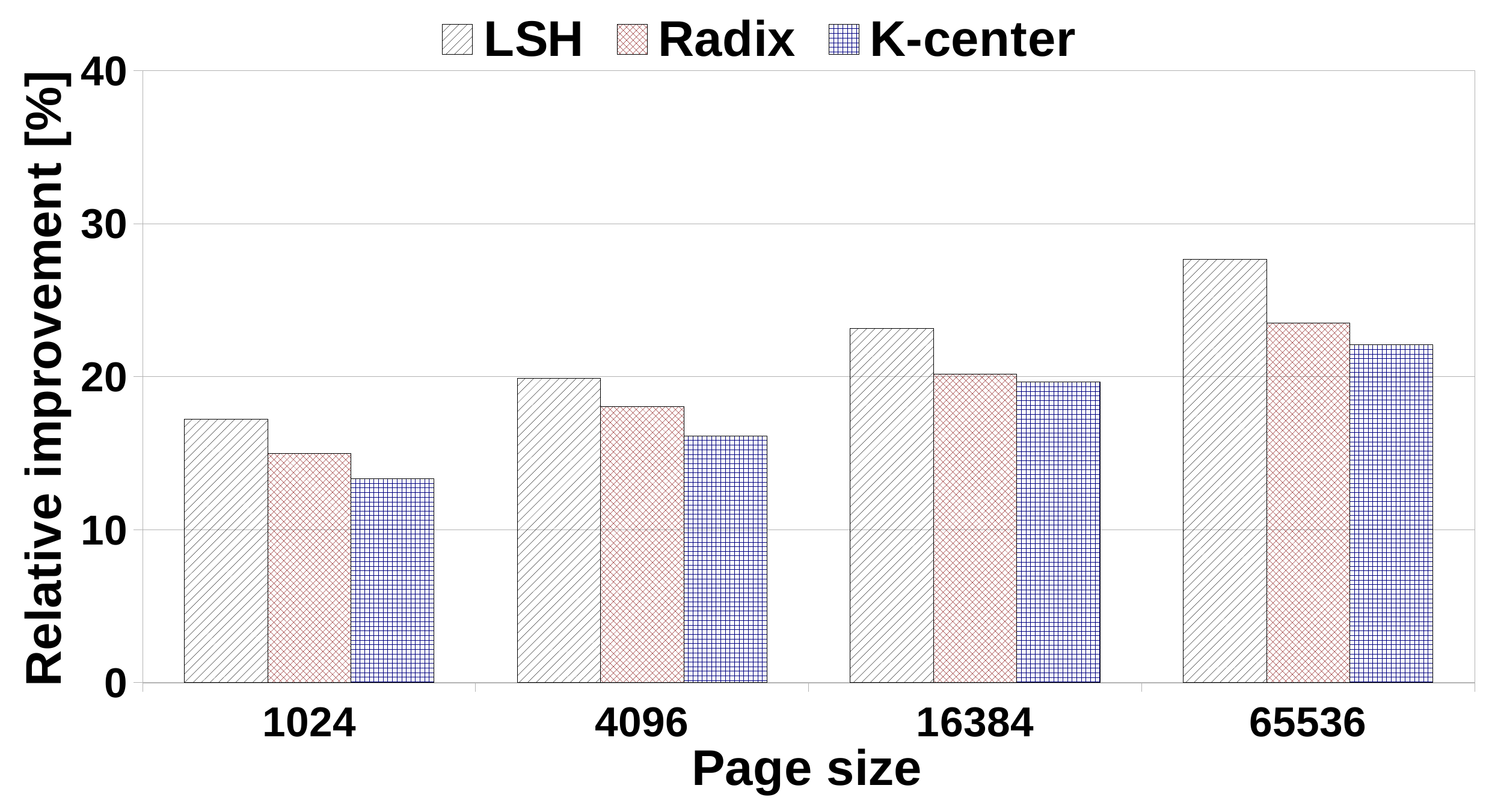}\label{fig:opt-page-miss}}
\caption{Page-level reordering heuristics comparison: (a) Execution time. (b) Relative improvement over basic LRU.}\label{fig:opt-chunk}
\end{center}
\end{figure*}

\subsection{Reordering Heuristics Comparison}\label{experiments:reordering}

We evaluate the performance of the proposed reordering heuristics -- locality-sensitive hashing (LSH), radix sort (Radix), and k-center clustering (K-center) -- as a function of the page size. We measure the reordering execution time and the relative improvement over vector-level basic LRU. In order to evaluate only the reordering effect, we do not include batching in these experiments. The results for \texttt{skewed} with memory budget of $1\%$ from the model size are depicted in Figure~\ref{fig:opt-chunk}.

\textbf{Reordering time.}
Figure~\ref{fig:opt-time-comp} shows the execution time of the reordering heuristics. As expected, when the page size, i.e., the number of vectors $\vec{u}_{i}$ considered together, increases, the reordering time increases for all the methods. Overall, Radix is the most scalable method. It executes in less than a second even for pages with $2^{16}$ vectors. K-center is infeasible for more than $1024$ vectors---the reordering time starts to dominate the processing time. While more scalable than K-center, LSH runs in more than $10$ seconds for $2^{16}$ vectors. Thus, from an execution time point-of-view, Radix is clearly the winner---it is faster than the others by $1$ to $3$ orders of magnitude.

\textbf{Improvement over basic LRU.}
We measure the number of page misses to model $V$ for each of the heuristics and compare against the number of page misses corresponding to the basic LRU replacement policy. Figure~\ref{fig:opt-page-miss} depicts the relative improvement. As expected, when the $U$ page size increases, the improvement over LRU increases since more vectors $\vec{u}_{i}$ are grouped together, thus, there are more opportunities for access sharing. K-center remains the worst method even in this category. However, LSH provides the largest reduction in page misses over LRU---almost $30\%$ for large page sizes. The reduction corresponding to Radix is around $20\%$.

Given that LSH has a much higher overhead and the relatively small reduction in page misses over Radix, we consider only Radix in further experiments. Moreover, LSH requires careful tuning of its many parameters in order to achieve these results. This is not the case for Radix which has no tunable parameter.

\begin{figure*}[htbp]
\begin{center}
\subfloat[]{\includegraphics[width=0.48\textwidth]{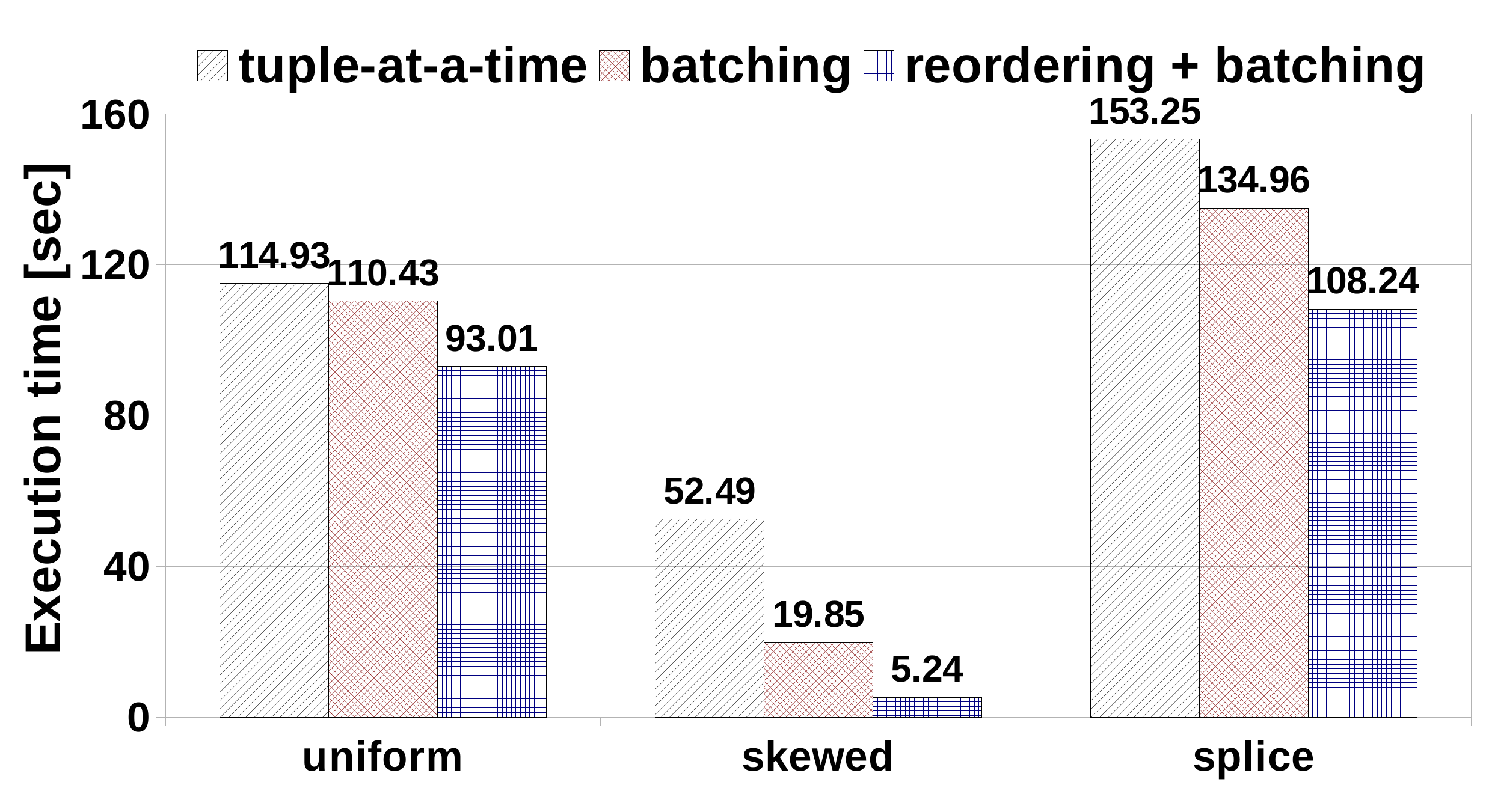}\label{fig:order-breakdown}}
\subfloat[]{\includegraphics[width=0.48\textwidth]{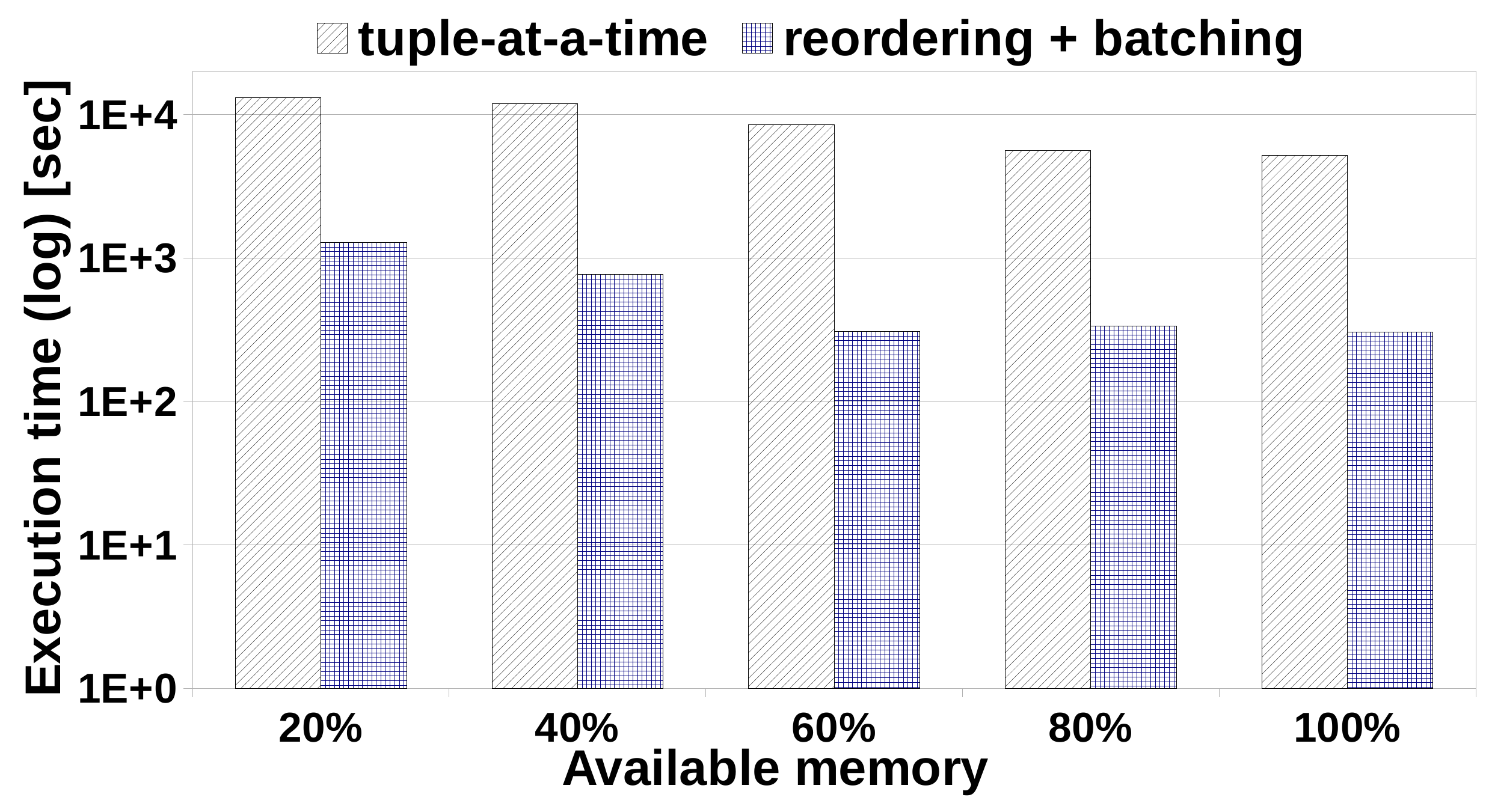}\label{fig:order-end2end}}
\caption{Dot-product join execution time: (a) Breakdown per page. (b) Over the full dataset.}\label{fig:dp-reorder-batch}
\end{center}
\end{figure*}

\subsection{Dot-Product Join Evaluation}

We evaluate the execution time of the complete dot-product join operator at page-level as well as over the full dataset. Figure~\ref{fig:order-breakdown} depicts the page-level improvement brought by batching and reordering over the naive tuple-at-a-time processing with LRU replacement. The memory budget is $20\%$ of the model size, while the page size is $4096$. Batching-only is executed over the arbitrary order of the $U$ vectors. Its benefits are considerably higher for the \texttt{skewed} dataset because of the larger overlap between narrow vectors---$300$ non-zero entries compared to $3000$ for \texttt{uniform}. Reordering reduces the execution time further since it allows for batches with larger sizes---the overhead of Radix reordering is included in the results. The one order of magnitude difference between tuple-at-a-time processing and dot-product join translates into an order of magnitude difference over the full dataset. Figure~\ref{fig:order-end2end} depicts the results for \texttt{skewed} as a function of the memory budget. When the available memory reaches a threshold at which the heavy accessed parts of the model can be buffered in memory, e.g., $60\%$ for \texttt{skewed}, there is little improvement with an additional increase.

\subsection{Dot-Product Join in SGD}\label{ssec:experiments:sgd}

Figure~\ref{fig:end-to-end} depicts the contribution of dot-product to the overall SGD computation for LR over \texttt{uniform} and \texttt{splice}, and for LMF over \texttt{matrix}, respectively. In addition to dot-product, SGD includes gradient computation and model update---which also requires secondary storage access. While we update the model for every example, the buffer manager is responsible for flushing dirty pages to secondary storage when memory is not available. As the memory budget for buffering the model increases, the dot-product and SGD time per iteration drop in all the cases. The relative contribution of dot-product to SGD is highly-sensitive to the memory budget and the characteristics of the dataset. For \texttt{uniform} (Figure~\ref{fig:uniform-end2end}), dot-product takes as little as one third of the SGD execution time at $60\%$ memory budget. In this case, the overhead of updating the model to storage dominates the execution time. For \texttt{splice} (Figure~\ref{fig:splice-end2end}), SGD is dominated by dot-product computation. This is because the model is considerably smaller, while the size of $U$ and the access sparsity are larger than for \texttt{uniform}. Due to the extreme model size, the only experiments we perform on \texttt{matrix} are for memory budgets of $10\%$ and $20\%$ (Figure~\ref{fig:lmf-end2end}). The results follow a similar trend.

\begin{figure*}[htbp]
\begin{center}
\subfloat[]{\includegraphics[width=0.33\textwidth]{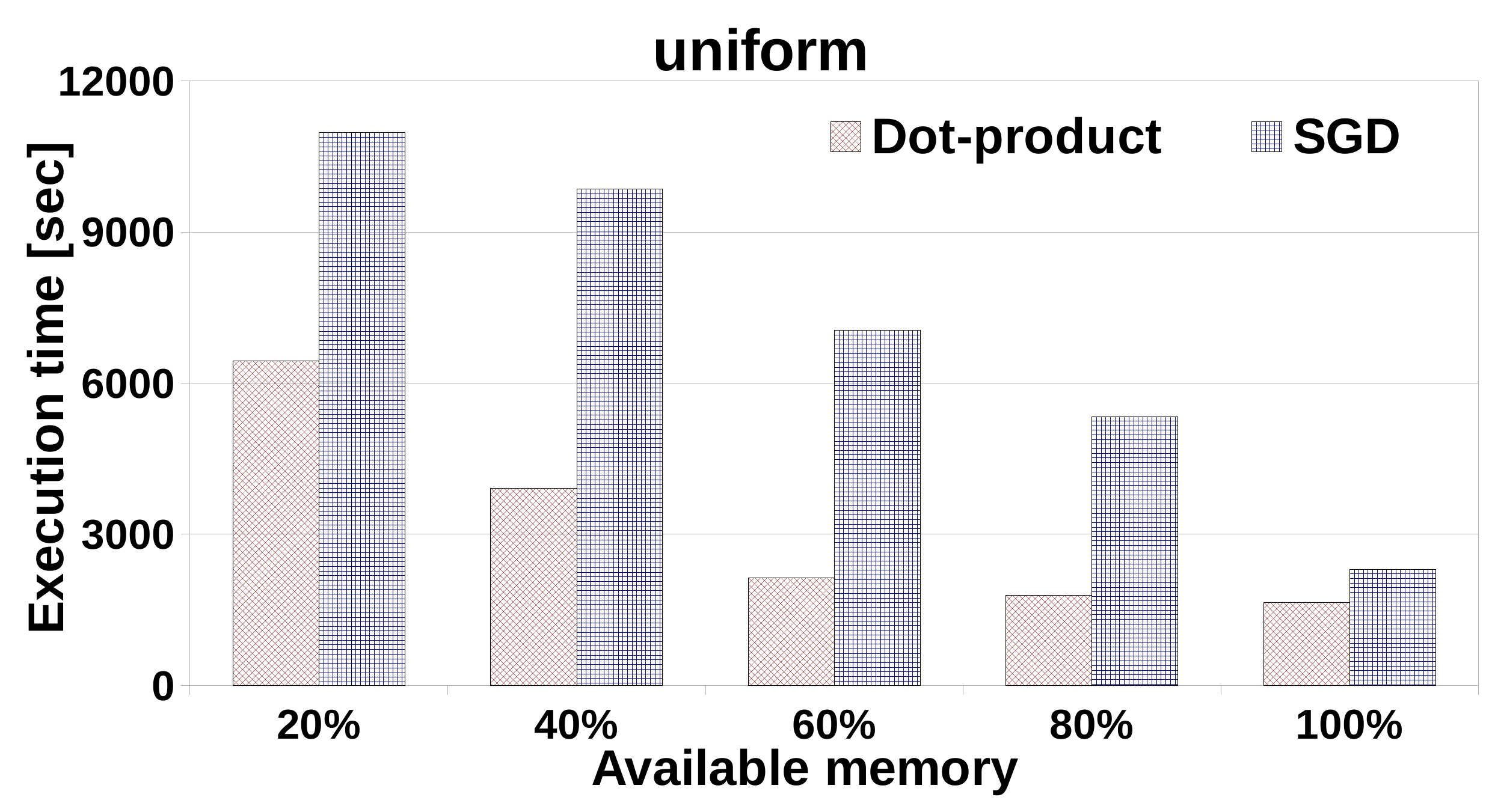}\label{fig:uniform-end2end}}
\subfloat[]{\includegraphics[width=0.33\textwidth]{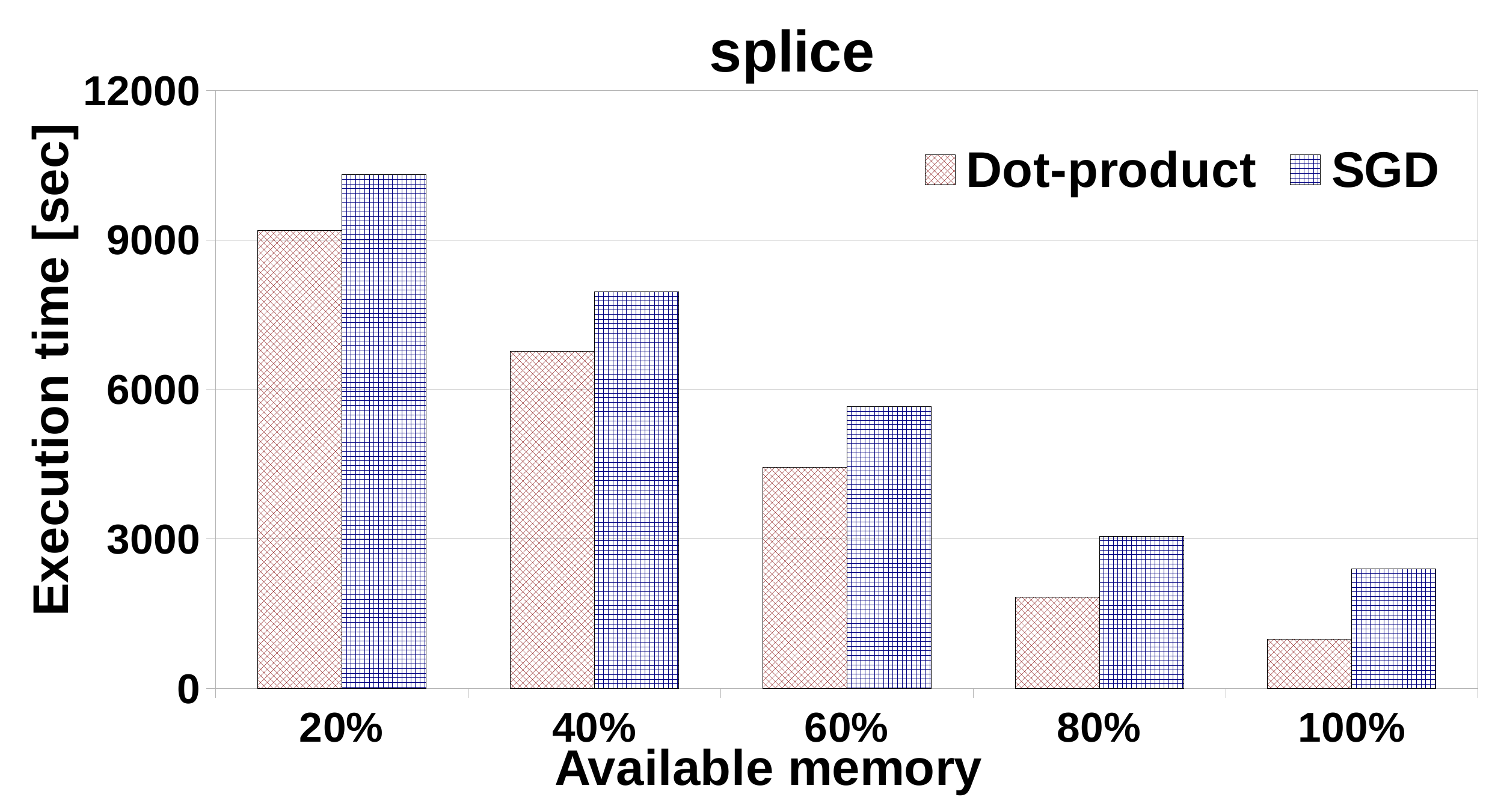}\label{fig:splice-end2end}}
\subfloat[]{\includegraphics[width=0.33\textwidth]{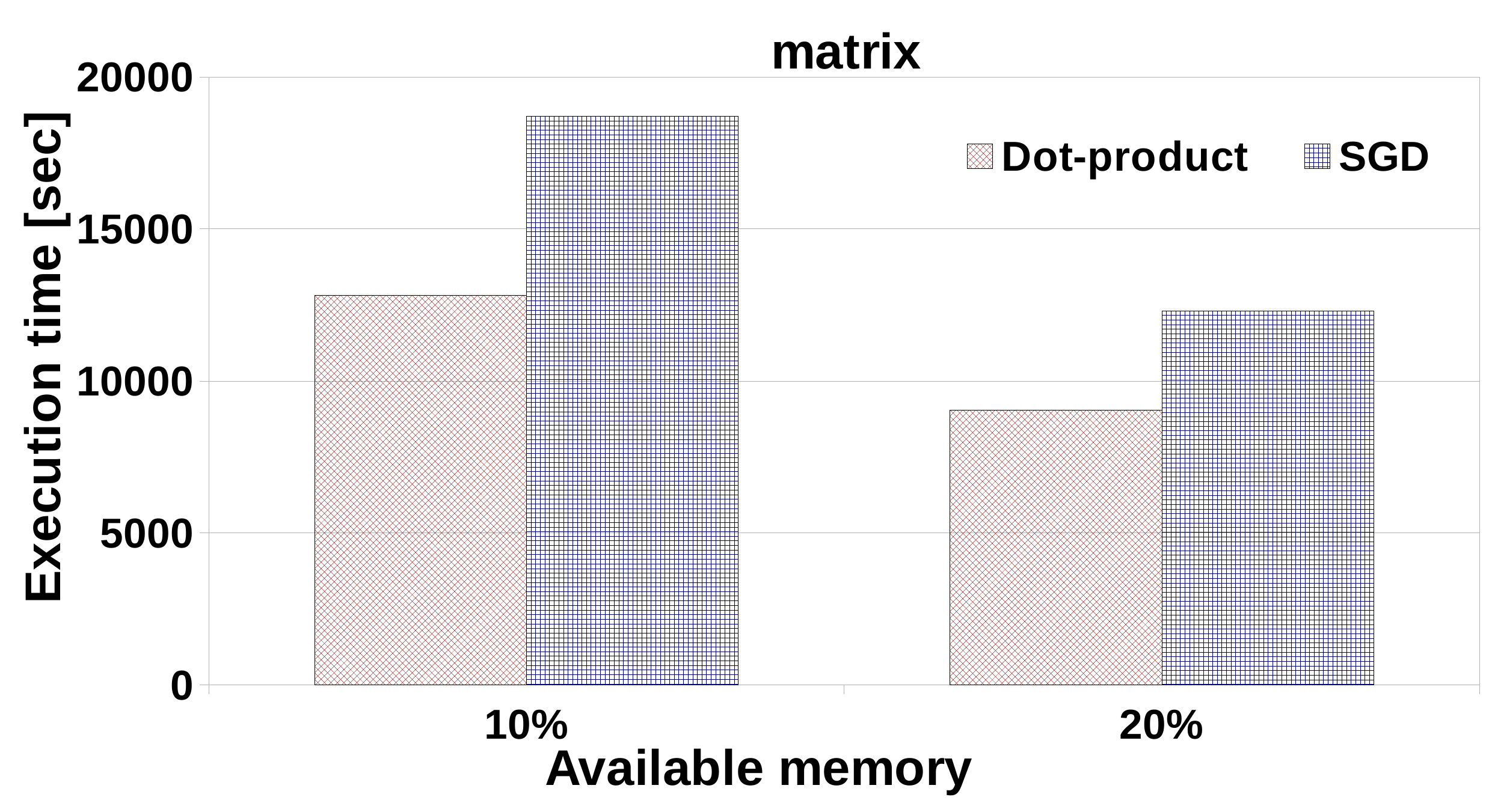}\label{fig:lmf-end2end}}
\caption{Dot-product and SGD execution time per iteration: (a) \texttt{uniform}; (b) \texttt{splice}; (c) \texttt{matrix}. While the model is updated for every example, not all dimensions get updated and not all updated dimensions are written back.}
\label{fig:end-to-end}
\end{center}
\end{figure*}

\subsection{Dot-Product Join vs. Alternatives}\label{ssec:experiments:comparison}

We validate the efficiency of the proposed dot-product join operator by comparing its execution time against that of the alternative solutions presented in Section~\ref{sec:baseline}. We have implemented relational (R) and \texttt{ARRAY}-relation in PostgreSQL (PG), relational in MonetDB~\cite{monetdb}, and a complete array solution in SciDB~\cite{scidb}. While PostgreSQL is a legacy database server with an extended set of features -- including support for \texttt{ARRAY} data type -- MonetDB is a modern column-store database designed for modern architectures, e.g., multi-core CPUs and large memory capacity. SciDB is a parallel database built on the array data model that has native support for linear algebra operations over massive arrays. This allows for the execution of the SpMV kernel as a simple function call independent of the size of the operands, i.e., out-of-core Matlab or R. The memory budget is set to $40\%$ of the model size across all the GLADE implementations. PostgreSQL is configured with 10 GB of buffer memory. Since MonetDB uses memory mapping for its internal data structures, it is not possible to limit its memory usage since system swapping is automatic. This gives a tremendous advantage to MonetDB when the model fits in memory. Indexes are built on the \texttt{index} attribute of all the model tables for PostgreSQL and MonetDB. The index building time is not included. SciDB is configured as a single server without any memory constraints---it can use all the available memory.

Table~\ref{tbl:comp} summarizes the execution time across all the datasets considered in the experiments. If the execution does not finish in 24 hours, we include N/A in the table. The proposed dot-product join operator is the only solution that finishes processing for all the datasets in the allocated time. It is also, in most cases, the fastest. The PostgreSQL and SciDB implementations are at the other extreme. For PG \texttt{ARRAY}-relation (PG A) and SciDB, only \texttt{MovieLens} finishes within 24 hours. In the case of PG A, the main reason for this is the immense size of the intermediate \texttt{ARRAY}-relation join over which the aggregation is computed. SciDB seems incapable to cope with large vectors having millions of dimensions. This is a problem not only for the SpMV kernel, but also for ingesting such highly-dimensional arrays. The PG relational solution is more efficient---10$\times$ faster on \texttt{MovieLens} and finishing \texttt{splice} in reasonable time. However, the model sizes of \texttt{MovieLens} and \texttt{splice} are the smallest among all datasets. MonetDB achieves decent runtime for all the datasets on the LR task. For the small model in \texttt{splice}, MonetDB finishes the fastest among all the solutions. This is because the computation can be executed completely in memory. When it comes to large models, while MonetDB is comparable with dot-product join on \texttt{uniform}, it is 6$\times$ slower on \texttt{skewed}. This shows that MonetDB is not able to achieve the gains of data reordering as the dot-product join operator does. For truly large models in LMF tasks, however, MonetDB fails to execute them efficiently because it materializes the immense intermediate join results. Since LMF requires a join between the $L$ matrix, the $R$ matrix, and the data (Section~\ref{ssec:problem:grad-descent}), even a small dataset such as \texttt{MovieLens} can result in an intermediate join size of $6K\times 4K\times 1000$ (192G) in the case of rank $1000$. Notice also that -- even though the relational solution generates competitive results in some cases -- it cannot generate the result in a non-blocking manner, as dot-product join does.


\begin{table*}[htbp]
  \begin{center}
    \begin{tabular}{l||r|rr|r|r|r}

      & GLADE R & PG R & PG A & MonetDB R & SciDB A & DOT-PRODUCT JOIN \\
	
    \hline
    \hline

    \texttt{uniform} & 37851 & N/A & N/A & 4280 & N/A & \textbf{3914} \\
	
    \texttt{skewed} & 13000 & N/A & N/A & 4952 & N/A & \textbf{786} \\
	
    \texttt{matrix} & N/A & N/A & N/A & N/A & N/A & \textbf{9045} \\

    \texttt{splice} & 9554 & 81054 & N/A & \textbf{2947} & N/A & 6769 \\
	
    \texttt{MovieLens} & 905 & 5648 & 72480 & N/A & 1116 & \textbf{477} \\
	
    \hline
    \end{tabular}
  \end{center}

\caption{Dot-product join execution time (in seconds) across several relational (R) and array (A) solutions. N/A stands for not finishing the execution in $24$ hours.}\label{tbl:comp}
\end{table*}

\subsection{Discussion}\label{ssec:experiments:discussion}

Our experiments identify Radix sort as the fastest reordering heuristic and the only one scalable to large batches. The improvement over tuple-at-a-time processing with basic LRU replacement is larger in this case. While LSH achieves the largest reduction in page misses, Radix is not far behind. The reduction in dot-product computation execution time is as much as an order of magnitude when reordering and batching are combined. Independently and when integrated in SGD, the dot-product join operator's performance degrades gracefully when the memory budget reduces below a threshold at which the model cannot be buffered in memory. Out of all the alternatives, the dot-product join operator is the only solution that can handle Big Models in a scalable fashion. The relational implementation in GLADE is -- in general -- an order of magnitude or more slower than dot-product join, while the relational and \texttt{ARRAY}-relation PostgreSQL versions and SciDB do not finish execution even after $24$ hours---except for small models.  MonetDB achieves efficient execution time for small and intermediate model sizes with a massive memory budget. However, it fails miserably on truly large models---the problem addressed by dot-product join.

\section{Related Work}\label{sec:rel-work}

\textbf{In-database analytics.}
There has been a sustained effort to add analytics functionality to traditional database servers over the past years. MADlib~\cite{mad-skills,madlib} is a library of analytics algorithms built on top of PostgreSQL. It includes gradient descent optimization functions for training generalized linear models such as LR and LMF~\cite{bismarck}. This is realized through the UDF-UDA extensions existent in PostgreSQL. The model is represented as the state of the UDA. It is memory-resident during an iteration and materialized as an array attribute across iterations. This is not possible for Big Model analytics because the model cannot fit in memory and PostgreSQL has strict limitations on the maximum size of an attribute. GLADE~\cite{igd-glade-ola} follows a similar approach. Distributed learning frameworks such as MLlib~\cite{mllib} and Vowpal Wabbit~\cite{vowpal-wabbit} represent the model as a program variable and allow the user to fully manage its materialization. Thus, they cannot handle Big Models directly. Since the vectors are memory-resident, the dot-product is computed by a simple invocation of \textit{Algorithm Dot-Product}. In the case of MADlib, this is done inside the UDA code. Computing dot-product and other linear algebra operations as UDAs is shown to be more efficient than the relational and stored procedure solutions for low-dimensional vectors in~\cite{ordonez:UDA}. SLACID~\cite{slacid} is a solution to implement sparse matrices and linear algebra operations inside a column-oriented in-memory database in which the emphasis is on the storage format. Adding support for efficient matrix operations in distributed frameworks such as Hadoop~\cite{Seo:Hama} and Spark~\cite{Lele:matrix-op} has been also investigated recently. We consider a more general problem -- dot-product is a sub-part of matrix multiplication -- in a centralized environment. Moreover, gradient descent optimization does not involve matrix operations---only dot-product.

\textbf{Big Model parallelism.}
Parameter Server~\cite{parameter-server-2,parameter-server} is the first system that addresses the Big Model analytics problem. Their approach is to partition the model across the distributed memory of multiple \textit{parameter servers}---in charge of managing the model. The training vectors $U$ are themselves partitioned over multiple workers. A worker iterates over its subset  of vectors and -- for each non-zero entry -- makes a request to the corresponding parameter server to retrieve the model entry. Once all the necessary model entries are received, the gradient is computed, the model is updated and then pushed back to the parameter servers. Instead of partitioning the model across machines, we use secondary storage. Minimizing the number of secondary storage accesses is the equivalent of minimizing network traffic in Parameter Server. Thus, the optimizations we propose -- reordering and batching -- are applicable in a distributed memory environment. They are not part of Parameter Server. In STRADS~\cite{Lee:STRADS}, parameter servers are driving the computation, not the workers. The servers select the subset of the model to be updated in an iteration by each worker. In our case, this corresponds to ignoring some of the non-zero entries in a vector $\vec{u}_{i}$ to further reduce I/O. We plan to explore strategies to select the optimal parameter subset in the future.

\textbf{Learning over normalized data.}
The integration of relational join with gradient, i.e., dot-product, computation has been studied in~\cite{sgd-over-join,sgd-over-join-2,bgd-over-factorized-join}. However, the assumption made in all these papers is that the vectors $\vec{u}_{i}$ are vertically partitioned along their dimensions. A join is required to put them together before computing the dot-product. In~\cite{sgd-over-join}, the dot-product computation is pushed inside the join and only applicable to BGD. The dot-product join operator adopts the same idea. However, this operator has to compute the dot-product which is still evaluated inside a UDA in~\cite{sgd-over-join}. The join is dropped altogether in~\cite{sgd-over-join-2} when similar convergence is obtained without considering the dimensions in an entire vertical partition. A solution particular to linear regression is shown to be efficient to compute when joining factorized tables in~\cite{bgd-over-factorized-join}. In all these solutions, the model is small enough to fit entirely in memory. Moreover, they work exclusively for BGD.

\textbf{Joins.}
Dot-product join is a novel type of join operator between an ARRAY attribute and a relation. We are not aware of any other database operator with the same functionality. From a relational perspective, dot-product join is most similar to index join~\cite{dbbook}. However, for every vector $\vec{u}_{i}$, we have to probe the index on model $V$ many times. Thus, the number of probes can be several orders of magnitude the size of $U$. The proposed techniques are specifically targeted at this scenario. The batched key access join\footnote{\url{https://dev.mysql.com/doc/refman/5.6/en/bnl-bka-optimization.html}} in MySQL is identical to our batching optimization applied at vector level. However, it handles a single probe per tuple and its reordering is aimed at generating sequential storage access for $V$. Array joins~\cite{Duggan:array-joins} are a new class of join operators for array databases. While it is possible to view dot-product join as an array join operator, the main difference is that we consider a relational system and push the aggregation inside the join. This avoids the materialization of the intermediate join result.

\textbf{SpMV kernel.}
As we have already discussed in the paper, the dot-product join operator is a constrained formulation of the standard sparse matrix vector (SpMV) multiplication problem. Specifically, the constraint imposes the update of the vector after each multiplication with a row in the sparse matrix. This makes impossible the direct application of SpMV kernels to Big Model analytics---beyond BGD. Moreover, we consider the case when the vector size goes beyond the available memory. SpMV is an exhaustively studied problem with applications to high performance computing, graph algorithms, and analytics. An extended discussion of the recent work on SpMV is presented in~\cite{ssd-spmv}---on which we draw in our discussion. \cite{spmv-multicore} and~\cite{phi-spmv} propose optimizations for SpMV in multicore architectures, while~\cite{spmv-partitioning} and~\cite{spmv-dist-partitioning} optimize distributed SpMV for large scale-free graphs with 2D partitioning to reduce communication between machines. Array databases such as SciDB~\cite{scidb} and Rasdaman~\cite{rasdaman} support SpMV as calls to optimized linear algebra libraries such as Intel MKL and Trilinos. There has also been preliminary research on accelerating matrix multiplication with GPUs~\cite{spmv-vldb} and SSDs~\cite{ssd-spmv}, showing that speedups are limited by I/O and setup costs. None of these works store the vector in secondary storage.

\section{Conclusions and Future Work}\label{sec:conclusions}

In this paper, we propose a database-centric solution to handle Big Models, where the model is offloaded to secondary storage. We design and implement the first dot-product join physical database operator that is optimized to execute secondary storage array-relation dot-products effectively. We prove that identifying the optimal access schedule for this operator is NP-hard. We propose dynamic batching and reordering techniques to minimize the overall number of secondary storage accesses. We design three reordering heuristics -- LSH, Radix, and K-center -- and evaluate them in terms of execution time and reordering quality. We experimentally identify Radix as the most scalable heuristic with significant improvement over basic LRU. The dot-product join operator's performance degrades gracefully when the memory budget reduces. Out of all the alternatives, the dot-product join operator is the only solution that can handle Big Models in a scalable fashion. Alternative solutions are -- in general -- an order of magnitude or more slower than dot-product join.

While the focus of the paper is on vector dot-product in the context of gradient descent optimization, the proposed solution is general enough to be applied to any SpMV kernel operating over highly-dimensional vectors. Essentially, the dot-product join operator is the first SpMV kernel that accesses the vector from secondary storage. Previous solutions have -- at most -- considered accessing the sparse matrix out-of-memory. From a database perspective, dot-product join is the first join operator between an array and a relation. Existing joins between an array attribute and a relation are tremendously ineffective. The experimental results in the paper confirm this. Thus, we find dot-product join extremely relevant in the context of polystores where multiple storage representations are combined inside a single system. Dot-product join eliminates the conversion between representations necessary to execute hybrid operations.

In future work, we plan to include data parallelism in the dot-product join operator by partitioning the array over multiple threads that update the model concurrently. While lack of parallelism may be seen as a limitation of the dot-product join operator -- especially in a Big Data setting -- we argue that addressing the sequential problem first is a mandatory requirement in order to design a truly scalable parallel solution. We also plan to investigate partition strategies from SpMV kernels to improve the disk access pattern to the vector further. Since SSD storage has improved dramatically over the last years, storing the matrix and the vector on SSDs is another topic for future research. Finally, we plan to expand the range of models covered by dot-product join beyond LR and LMF. As long as a dot-product operation is required between highly-dimensional vectors, the dot-product join operator is a solution to consider.

{\small
\bibliographystyle{abbrv}

}

\end{document}